\documentclass[a4paper,UKenglish,cleveref, autoref]{lipics-v2019}

\usepackage{amsmath,amsthm,amssymb}
\usepackage[T1]{fontenc}
\usepackage{subfiles}
\usepackage{xspace}
\usepackage[dvipsnames]{xcolor}
\usepackage{hyperref}
\usepackage{mathtools}
\usepackage{cite}

\theoremstyle{plain}

\definecolor{ablue}{rgb}{0.3,0.4,0.8}
\definecolor{ared}{rgb}{0.95,0.4,0.4}
\definecolor{agreen}{rgb}{0,0.5,0.25}
\definecolor{ayellow}{rgb}{0.95,0.85,0.3}

\DeclarePairedDelimiter{\oiv}{(}{)}

\DeclarePairedDelimiter{\roiv}{[}{)}

\newcommand{\floor}[1]{\left\lfloor #1 \right\rfloor}

\newcommand{\Reals}{\mathbb{R}}
\newcommand{\Nats}{\mathbb{N}}
\newcommand{\Ints}{\mathbb{Z}}

\newcommand{\bi}{\mathbf{i}}

\newcommand{\cG}{\mathcal{G}}
\newcommand{\cC}{\mathcal{C}}

\newcommand{\sig}{\sigma}
\newcommand{\eps}{\varepsilon}
\renewcommand{\to}{\rightarrow}

\newcommand{\EC}{\mathcal{EC}}
\newcommand{\BEC}{\mathcal{BEC}}

\newcommand\eqdef{\mathrel{\overset{\makebox[0pt]{\mbox{\normalfont\tiny\sffamily def}}}{=}}}

\DeclareMathOperator{\diam}{diam}
\DeclareMathOperator{\radius}{radius}
\newcommand{\myin}{\mathrm{in}}
\newcommand{\myout}{\mathrm{out}}

\newcommand{\Bin}{B_{\myin}}
\newcommand{\Bout}{B_{\myout}}

\DeclareMathOperator{\poly}{poly}

\DeclareMathOperator{\Vol}{Vol}

\newcommand{\Var}{\mathsf{Var}}
\newcommand{\Cla}{\mathsf{Cla}}

\renewcommand{\leq}{\leqslant}
\renewcommand{\geq}{\geqslant}

\newcommand{\etal}{\emph{et~al.}}

\newcommand{\IS}{\textsc{Independent Set}\xspace}

\title{How does object fatness impact the complexity of packing in $d$ dimensions?}

\titlerunning{How does object fatness impact packing?}

\author{S\'andor Kisfaludi-Bak}{Max Planck Institut f\"ur Infromatik, Saarbr\"ucken, Germany}{sandor.kisfaludi-bak@mpi-inf.mpg.de}{}{}

\author{D\'aniel Marx}{Max Planck Institut f\"ur Infromatik, Saarbr\"ucken, Germany}{dmarx@cs.bme.hu}{}{Supported by ERC Consolidator Grant SYSTEMATICGRAPH (No.~{725978})}
\author{Tom C. van der Zanden}{Department of Data Analytics and Digitalisation, Maastricht University, The Netherlands}{T.vanderZanden@maastrichtuniversity.nl}{}{}

\authorrunning{S\'andor Kisfaludi-Bak, D\'aniel Marx, Tom C. van der Zanden}

\Copyright{S\'andor Kisfaludi-Bak, D\'aniel Marx, and Tom C. van der Zanden}



\nolinenumbers
\hideLIPIcs

\EventEditors{Pinyan Lu and Guochuan Zhang}
\EventNoEds{2}
\EventLongTitle{30th International Symposium on Algorithms and Computation (ISAAC 2019)}
\EventShortTitle{ISAAC 2019}
\EventAcronym{ISAAC}
\EventYear{2019}
\EventDate{December 8--11, 2019}
\EventLocation{Shanghai University of Finance and Economics, Shanghai, China}
\EventLogo{}
\SeriesVolume{149}
\ArticleNo{39}

\begin{document}

\maketitle
\begin{abstract}
Packing is a classical problem where one is given a set of subsets of Euclidean space called objects, and the goal is to find a maximum size subset of objects that are pairwise non-intersecting. The problem is also known as the Independent Set problem on the intersection graph defined by the objects. Although the problem is NP-complete, there are several subexponential algorithms in the literature. One of the key assumptions of such algorithms has been that the objects are fat, with a few exceptions in two dimensions; for example, the packing problem of a set of polygons in the plane surprisingly admits a subexponential algorithm. In this paper we give tight running time bounds for packing similarly-sized non-fat objects in higher dimensions.

We propose an alternative and very weak measure of fatness called the stabbing number, and show that the packing problem in Euclidean space of constant dimension $d \geq 3$ for a family of similarly sized objects with stabbing number $\alpha$ can be solved in $2^{O(n^{1-1/d}\alpha)}$ time. We prove that even in the case of axis-parallel boxes of fixed shape, there is no $2^{o(n^{1-1/d}\alpha)}$ algorithm under ETH. This result smoothly bridges the whole range of having constant-fat objects on one extreme ($\alpha=1$) and a subexponential algorithm of the usual running time, and having very ``skinny'' objects on the other extreme ($\alpha=n^{1/d}$), where we cannot hope to improve upon the brute force running time of $2^{O(n)}$, and thereby characterizes the impact of fatness on the complexity of packing in case of similarly sized objects. We also study the same problem when parameterized by the solution size $k$, and give a $n^{O(k^{1-1/d}\alpha)}$ algorithm, with an almost matching lower bound: there is no algorithm with running time of the form $f(k)n^{o(k^{1-1/d}\alpha/\log k)}$ under ETH. One of our main tools in these reductions is a new wiring theorem that may be of independent interest. 
\end{abstract}

\section{Introduction}


Many well-known NP-hard problems (e.g. {\sc Independent Set,
  Hamilton Cycle, Dominating Set}) can be solved in time
$2^{O(\sqrt{n})}$ when restricted to planar graphs, while only
$2^{O(n)}$ algorithms are known for general graphs \cite{DBLP:conf/soda/KleinM14,DBLP:conf/stacs/PilipczukPSL13,DBLP:journals/ipl/FominLRS11,DBLP:conf/esa/Thilikos11,DBLP:journals/csr/DornFT08,DemaineFHT05,DBLP:journals/jcss/DornFT12,DornPBF10,DBLP:journals/siamcomp/FominT06,DBLP:journals/talg/DemaineFHT05}. This beneficial
effect of planarity is known as the ``square root phenomenon,'' and
can be exploited also in the context of 2-dimensional geometric
problems where the problem is defined on various intersection graphs
in $\Reals^2$ \cite{AlberF04,BasteT17,FominLS12,MarxP15}. In particular, consider the geometric packing problem
where, given a set of polygons in $\Reals^2$, the task is to find a
subset of $k$ pairwise disjoint polygons. This problem
can be solved in time $n^{O(\sqrt{k})}$ \cite{MarxP15}, which -- when expressed only a as a function of the input -- gives an
$n^{O(\sqrt{n})}=2^{O(\sqrt{n}\log n)}$ algorithm for finding a maximum size
disjoint subset.

Can these 2-dimensional subexponential algorithms be generalized to higher dimensions? It seems that the natural generalization is to aim for $2^{O(n^{1-1/d})}$ , or in case of parameterized problems, either $2^{O(k^{1-1/d})}\cdot n^{O(1)}$ or $n^{O(k^{1-1/d})}$ time algorithms in $d$-dimensions: the literature contains upper and lower bounds of this form (although sometimes with extra logarithmic factors in the exponent) \cite{BergBKMZ18,MarxS14,SmithW98}. 
However, all of these algorithms
have various restrictions on the object family on which the intersection graph
is based: there is no known analogue of the $n^{O(\sqrt{k})}$ time algorithm of
Marx and Pilipczuk \cite{MarxP15} in higher dimensions with the same
generality of objects. There is a good reason for this: it is easy to see that
any $n$-vertex graph can be expressed as the intersection graph of
3-dimensional simple polyhedra. Thus a subexponential algorithm for
3-dimensional objects without any severe restriction would give a
subexponential algorithm for \textsc{Independent Set} on general graphs,
violating standard complexity-theoretic assumptions.

What could be reasonable restrictions on the objects that allow
running times of the form, e.g., $2^{O(n^{1-1/d})}$? One of the most
common restrictions is to study a set $F \subset 2^{\Reals^d}$ of
\emph{fat objects}, where for each object $o \in F$ the ratio
$\radius(\Bin(o))/\radius(\Bout(o))$ is at least some fixed positive
constant.  (We denote by $\radius(\Bin)$ and $\radius(\Bout)$ the
radius of the inscribed and circumscribed ball respectively.) Another
common restriction is to have \emph{similarly sized} objects, that is,
a family $F$ where the ratio of the largest and smallest object
diameter is at most some absolute constant. Many results concern only
\emph{unit disk graphs}, where~$F$ consists of unit disks in the
plane: unit disks are both fat and similarly sized. The focus of our
paper is to explore the role of fatness in the context of packing
problems and to understand when and to what extent fatness decreases
the complexity of the problem.
We observe that fatness is a crucial requirement for subexponential algorithms in higher dimensions, and this prompts us to explore in a quantitave way how fatness influences the running time. For this purpose, we introduce a parameter $\alpha$ describing the fatness of the objects and give upper and lower bounds taking into account this parameter as well.

More precisely, we introduce the notion of the \emph{stabbing number}, which
can be regarded as an alternative measure of fatness. This slightly extends a
similar definition by Chan~\cite{Chan03}. We say that an object $o$ is stabbed
by a point $p$ if $p\in o$.  A family of objects $F\subseteq 2^{\Reals^d}$ is
$\alpha$\emph{-stabbed} if for any $r\in \Reals$, the subset of $F$-objects
$o$ of diameter $\diam(o)\in \roiv{r/2,r}$ contained in any ball of radius $r$
can be stabbed by $\alpha^d$ points. The stabbing number of $F$ is defined as
$\inf_{\alpha \in \roiv{1,\infty}} \{F \text{ is   $\alpha$-stabbed}\}$. Note
that a set of $n$ objects in $d$-dimensions has stabbing number at most
$n^{1/d}$. The stabbing number is closely related to the inverse of a common
measure of fatness. This relationship is explored in Section~\ref{sec:app_fatness}.



By adapting a separator theorem from~\cite{BergBKMZ18}, we can give an
algorithm where the running time smoothly goes from $2^{O(n^{1-1/d})}$ to
$2^{O(n)}$ as the stabbing number goes from $O(1)$ to the maximum possible
$n^{1/d}$.

\begin{theorem}\label{thm:weightedalg}
Let $\alpha\in \roiv{1,\infty}$ and $2\leq d \in \Nats$ be fixed constants.
There is an algorithm that solves \IS for intersection graphs of
similarly sized $\alpha$-stabbed objects in $\Reals^d$ running in time
$2^{O(n^{1-1/d}\alpha)}$.
\end{theorem}

As mentioned, the stabbing number is at most $n^{1/d}$, and this
algorithm runs in subexponential time whenever the stabbing number is
better than this trivial upper bound, that is, whenever
$\alpha=o(n^{1/d})$ holds.


In order to have definite answers to the best running times achievable, we
also need a lower bound framework. A popular starting point in the past
decades is the Exponential Time Hypothesis (ETH)~\cite{ImpagliazzoP01}, which
posits that there exists a constant $\gamma>0$ such that there is no $2^{\gamma n}$ algorithm for the $3$-SAT problem. Classical NP-hardness reductions automatically yield
quantitative lower bounds on the running time under ETH. If enough care is taken to ensure that the constructed instance is sufficiently small, then one can find lower bounds that
match the best known algorithms~\cite{fptbook}. For the \IS problem, a lower
bound of $2^{\Omega(n)}$ is a consequence of classical reductions under ETH.

A standard way to explore the impact of a parameter such as fatness is to give
an algorithm where the parameter appears in the running time, together with a matching lower bound. However, the notion of ``matching lower bound'' needs to be defined precisely if we are expressing the running time as a function of two parameters, the size $n$ of the instance and the stabbing number $\alpha$ of the objects. 

A recent example of such an algorithm and lower bound involving two parameters is the paper by
Bir\'o~\etal~\cite{BiroBMMR17}, where it is shown that the coloring problem of
unit disk graphs with $\ell=n^\lambda$ colors can be solved in
$2^{O(\sqrt{n\ell}\log n)}$ time, where $\lambda\in [0,1]$ is a fixed
constant, and they also exclude algorithms of running time
$2^{o(\sqrt{n\ell})}$ under ETH. This is interesting since this smoothly
bridges the gap between a standard ``square root phenomenon'' algorithm ($\ell
= O(1))$) on one extreme and the brute force $2^{O(n)}$ on the other
($\ell=n^{1-o(1)}$). Our results show a similar behavior in the context of
fatness and the packing problem: the running time of
Theorem~\ref{thm:weightedalg} is optimal, with the
running time smoothly going from $2^{O(n^{1-1/d})}$ time 
in the case of $\alpha=O(1)$ to the trivial $2^{O(n)}$ time
of brute force when $\alpha=n^{1/d}$.

Let $\cG(d, L)$ denote the set of intersection graphs in $\Reals^d$ where each
object is an axis-parallel box whose side lengths form the multiset
$\{1,\dots,1,L\}$. Let us call such an axis-parallel box \emph{canonical}.
As usual, $n$ denotes the number of objects (the number of vertices in the graph).

For example, it is easy to see that $1\times 1\times L$ boxes have stabbing number
$O(L^{2/3})$. Any collection of $1\times 1\times L$ boxes of the same
orientation can be stabbed by the lattice generated by the vertices of such a
box, which has $O(L^2)$ points in a ball of radius $O(L)$. By taking the same
lattice for the two other orientations, we obtain a complete stabbing set of
size $O(L^2)$ inside a ball of radius $O(L)$ for all axis-parallel boxes of
this shape.
In general for $d\geq 3$, the stabbing number for canonical boxes is $\alpha=O(L^{1-1/d})$, so in
particular, for $L=1$ we have $\alpha = O(1)$, and for $L\geq n^{1/(d-1)}$
we have $\alpha = O(n^{1/d})$. In our main contribution, we show that this very restricted set of
non-fat objects is sufficient to prove the desired lower bound.

\begin{theorem}\label{thm:generallower}
Let $d\geq 3$ be fixed. Then there is a constant $\gamma>0$ such that for all $\alpha\in[1,n^{1/d}]$ it holds that
\textsc{Independent Set} on intersection graphs of $d$-dimensional canonical
axis-parallel boxes of stabbing number $\alpha$  has no algorithm running in time
$2^{\gamma n^{1-1/d}\alpha}$, unless ETH fails.
\end{theorem}

An immediate corollary is that the $2^{O(n)}$ time brute-force algorithm
cannot be improved, even for the intersection graph of axis-parallel boxes. This Corollary~\ref{cor:general} can also be derived from a simpler
construction by Chleb\'ik and
Chleb\'ikov\'a~\cite{DBLP:conf/soda/ChlebikC05}.

\begin{corollary}\label{cor:general}
Let $3\leq d \in \Nats$ be fixed.  Then \textsc{Independent Set} on
intersection graphs of axis-parallel boxes in $d$-dimensions has no algorithm
running in time $2^{o(n)}$, unless ETH fails.
\end{corollary}

In unit ball graphs, there is a lower bound of $2^{\Omega(n^{1-1/d})}$ under
ETH, which of course carries over to intersection graphs of fat
objects~\cite{BergBKMZ18}. This latter reduction is based on establishing
efficient routing constructions (called the ``Cube Wiring theorem'') in the
$d$-dimensional Euclidean grid. The crucial insight of the present paper is
that tight lower bounds for {\em nonfat} objects can be obtained via \IS on
induced subgraphs of the $d$-dimensional {\em blown-up} grid cube, where each
vertex is replaced by a clique of $t$ vertices, fully connected to the
adjacent cliques in all $d$ directions. First we establish a lower bound for
\IS on  subgraphs of such cubes (even for subgraphs of maximum degree 3),
using and extending the Cube Wiring theorem~\cite{BergBKMZ18}. Unlike for unit
balls, it now seems difficult to realize every such  subgraph $G$ as
intersection graph of appropriate boxes. Instead, we realize a graph $G'$ that
is obtained from $G$ by some number of double subdivisions (subdividing some
edge twice). As every double subdivision is known to increase the size of the
maximum independent set by exactly 1, switching to $G'$ does not cause a
problem in the reduction.

The key insight of the reduction (in 3-dimensions) is that if $t=L^2$,
then $t$ vertices can be represented with $1\times 1 \times L$ size
boxes arranged in an $L\times L$ grid, occupying
$O(L)\times O(L)\times O(L)$ space. Each $t$-clique of the blown-up
cube is represented by such arrangements of boxes. The main
challenge that we have to overcome is that the subgraph $G$ may
contain an arbitrary matching between two adjacent $t$-cliques. Given
two sets of $1\times 1 \times L$ size boxes arranged in two
$L\times L$ grids, it seems unclear whether such arbitrary connections
can be realized while staying in an $O(L)\times O(L) \times
O(L)$ region of space. However, we show that this is possible, as the $L\times L$ grid
arrangement allows easy reordering within the rows or within the
columns, and it is known that any permutation of a grid can be
obtained as doing a permutation first within the rows, then within the
columns, and finally one more time within the rows. Thus with some
effort, it is possible to build gadgets representing $L\times L$
vertices in an $O(L)\times O(L)\times O(L)$ region of space that allows arbitrary
matchings to be realized with the adjacent gadgets.

The idea is similar in higher dimensions $d>3$. We reduce from the \IS
problem on a subgraph of the blow-up of a $d$-dimensional grid where
each vertex is blown-up into a clique of $L^{d-1}$ vertices. Each
gadget now contains $L^{d-1}$ boxes of size
$1\times 1 \times \dots \times 1 \times L$ arranged in a grid. In order
to implement arbitrary matchings between adjacent gadgets, we
decompose every permutation of the $(d-1)$-dimensional grid into $O(d)$ simpler
permutations that are easy to realize in $d$-dimensional space.

We also study the complexity of packing in the context of parameterized
algorithms: the question is how much one can improve the brute force
$n^{O(k)}$ algorithm for finding $k$ independent objects. We present
a counterpart of Theorem~\ref{thm:weightedalg} in this setting.
\begin{theorem}\label{thm:paramalg}
Let $\alpha\in \roiv{1,\infty}$ and $2\leq d \in \Nats$. There is a
parameterized algorithm that solves independent set for intersection graphs of
similarly sized $\alpha$-stabbed objects in $\Reals^d$ running in time
$n^{O(k^{1-1/d}\alpha)}$, where the parameter $k$ is the size of the maximum
independent set.
\end{theorem}

If one regards the parameterized algorithm's running time in terms of
the instance size only, the result would be a
$2^{O(n^{1-1/d}(\log n)\alpha)}$ algorithm, which is slower than the
running time $2^{O(n^{1-1/d}\alpha)}$ provided by the latter
algorithm.  The parameterized algorithm is based on a separator
theorem by Miller~\etal~\cite{MillerTTV97}.

Finally, we sketch how the lower bound construction of Theorem~\ref{thm:generallower} can be adapted to a parameterized
setting, and obtain the following theorem:

\begin{theorem}\label{thm:paramlower}
Let $3\leq d \in \Nats$ be fixed. Then there is a constant $\gamma>0$ such that for all $\alpha\in[1,n^{1/d}]$ it holds that deciding if there is an 
independent set of size $k$ in intersection graphs of $d$-dimensional canonical
axis-parallel boxes of stabbing number $\alpha$  has no $f(k)n^{\gamma k^{1-1/d}\alpha/\log k}$
algorithm for any computable function $f$, unless ETH fails.
\end{theorem}

The crucial difference is that we are reducing from the \textsc{Partitioned
Subgraph Isomorphism} problem instead of \IS, which means that instead of
choosing or not choosing a box (representing choosing or not choosing a vertex
in the \IS problem), the solution needs to choose one of $n$ very similar
boxes (representing the choice of one of $n$ vertices in a class of the
partition). The overall structure  of the reduction (e.g., routing in
the blown-up $d$-dimensional grid) is similar to the proof of
Theorem~\ref{thm:generallower}.

\noindent\emph{Organization.} In Section~\ref{sec:app_fatness} we establish some bounds that relate the stabbing number to fatness. Section~\ref{sec:alg} presents both our non-parameterized and parameterized algorithm. In Section~\ref{sec:wire} we prove the wiring theorem that is necessary for both of our lower bounds. Sections~\ref{sec:lower} and~\ref{sec:app_paramlower} contain our lower bounds for the non-parameterized and parameterized problem respectively. Finally, Section~\ref{sec:conclusion} draws some conclusions and proposes two open problems.

\section{The relationship between the stabbing number and fatness}\label{sec:app_fatness}

In the usual definition of fatness, an object $o\subset \Reals^d$ is $\alpha$-fat if there exists a ball of radius $\rho_\myin$ contained in $o$ and a ball of radius $\rho_\myout$ that contains $o$, where $\rho_\myin/\rho_\myout = \alpha$. For a fixed constant $\alpha$ this is a useful definition and unifies many other similar notions in case of convex objects, i.e., it holds that a set of convex objects that is constant-fat for this notion of fatness are constant-fat for more restrictive definitions and vice versa. For our purposes however this definition is not fine-grained enough in the following sense. The fatness of a $1\times 1 \times n$ box in three dimensions would be $\Theta(n)$, just as the fatness of a $1\times n \times n$ box. As it will be apparent in what follows, we need a fatness definition according to which $1\times n \times n$ boxes are much more fat than $1\times 1 \times n$ boxes. For this purpose, we use the following weaker definition of fatness, that tracks the volume compared to a circumscribed ball more closely. (Note that constant-fat objects are also weakly constant-fat.)

\begin{definition}[Weakly $\alpha$-fat]
A measurable object $o\subseteq \Reals^d$ is $\alpha$-fat for some $\alpha\in \roiv{1,\infty}$ if $Vol(o)/Vol(B)\leq \alpha^d$, where $Vol(o)$ and $Vol(B)$ denotes the volume of $o$ and the volume of its circumscribed ball $B$ respectively.
\end{definition}

An object $o$ is \emph{strongly $\alpha$-fat} if for any ball $B$ centered inside $o$ we have $Vol(B\cap o)/Vol(B)\geq \alpha^d$. In case of convex objects, weak fatness coincides with strong fatness up to constant factors, see~\cite{StappenHO93}.

The next theorem shows that the inverse of the weak fatness of an object family is related to the stabbing number. In a sense, this means that the stabbing number is a further weakening of weak fatness. Note that in our setting, the stabbing number will be polynomial in $n$ (i.e., $\alpha = n^{\lambda}$ for some constant $\lambda$), so the $\log n$ term is insignificant. 

\begin{theorem}
Let $d$ be a fixed constant. Then the stabbing number of any family of $n$
weakly $(1/\alpha)$-fat (measurable) objects in $\Reals^d$  is $O(\alpha \log^{1/d} n)$.
\end{theorem}

\begin{proof}
Consider a family $F$ of weakly $1/\alpha$-fat objects. Let $B$ be a ball of
radius $\delta$, and let $F_B$ be the set of objects contained in $B$ of
diameter at least $\delta/2$. It is sufficient to show that we can stab $F_B$
with $O(\alpha^d\log n)$ points. Pick $k=\floor{(4\alpha)^d(\log n+1)}$ points
$p_1,\dots p_k$ independently uniformly at random in $B$. For any given object
$o$, its volume is at least $\Vol(B)/(4\alpha)^d$, so the probability that a
given $p_i$ is not in $o$ is at most $1-1/(4\alpha)^d$. Since the $k$ points are
chosen independently, the probability that a given object $o$ is unstabbed is at most
$\left(1-\frac{1}{(4\alpha)^d}\right)^k$. By the union
bound, the probability that there is an unstabbed object is at most

\[n\left(1-\frac{1}{(4\alpha)^d}\right)^k
= n\left(1-\frac{1}{(4\alpha)^d}\right)^{\floor{(4\alpha)^d(\log n+1)}}
< n(1/e)^{\log n+1}<1.\]

Consequently, there exists an outcome where all objects are stabbed.
\end{proof}

We conclude this section with the following theorem, which shows an even stronger
connection between fatness and stabbing in case of convex objects. The theorem
uses the existence of the John ellipsoid~\cite{John2014} and
the $\eps$-net theorem~\cite{Haussler1987}.

\begin{theorem}\label{thm:convexstabbed}
Let $d$ be a fixed constant. Then the stabbing number of any family of $n$
weakly $(1/\alpha)$-fat convex objects in $\Reals^d$  is $O(\alpha
\log^{1/d} \alpha)$.
\end{theorem}

\begin{proof}
Consider a family $F$ of weakly $1/\alpha$-fat convex objects. Let $B$ be a
ball of radius $\delta$, and let $F_B$ be the set of objects contained in $B$
of diameter at least $\delta/2$. It is sufficient to show that we can stab
$F_B$ with $O(\alpha^d\log \alpha)$ points. For any given object $o$, its
volume is at least $\Vol(B)/(4\alpha)^d$. Every convex object $o\in F_B$
contains an ellipsoid $\ell(o)\subseteq o$ such that $\Vol(o)/\Vol(\ell(o)) \leq
d^d$~\cite{John2014}. Since the VC-dimension of ellipsoids in $\Reals^d$ is
$O(d^2)$~\cite{DBLP:journals/corr/abs-1109-4347}, the $\epsilon$-net
theorem~\cite{Haussler1987} implies that the ellipsoids $\ell(o)\; (o\in F_B)$
can be stabbed by $O(\frac{d^2}{1/(4\alpha)^d}\log \frac{d^2}{1/(4\alpha)^d})
= O(\alpha^d\log \alpha)$ points. Since the ellipsoids are contained in their
respective objects, this point set also stabs all objects in $F_B$.
\end{proof}

\section{Algorithms}\label{sec:alg}

We require very little from the objects that we use in our algorithms. It is necessary that we can decide in polynomial time whether a point is contained in an object, whether two objects intersect, and whether an object intersects some given sphere, ball, and empty or dense hypercube. Let us assume that such operations are possible from now on.

\subsection{An algorithm with weighted cliques}

The algorithm for Theorem~\ref{thm:weightedalg} is an adaptation of the \IS algorithm for fat objects from~\cite{BergBKMZ18}, based on weighted cliques.

\begin{proof}[Proof of Theorem~\ref{thm:weightedalg}]
The algorithm works by finding a balanced separator of the objects, such that
the separator itself can be partitioned into cliques and this partition has
the property that the number of independent sets within the separator is
$2^{O(n^{1-1/d}\alpha)}$. The result then follows from applying this algorithm
recursively. Thus, we are left with the task to prove the existence of such a
separator.

We begin by picking a minimum size hypercube $H_0$ that contains at least
$n/(6^d + 1)$ objects, and we translate and scale everything
so that $H_0$ becomes a unit hypercube centered at the origin. We now define
$n^{1/d}$ hypercubes $H_1,\ldots,H_{n^{1/d}}$, which will be our candidate
separators. Each hypercube $H_i$ is centered at the origin and has edge length
$1 + \frac{2i}{n^{1/d}}$.

Each hypercube $H_i$ corresponds to a separator as follows: the separator
consists of the objects intersected by the boundary of the hypercube, and
separates the objects contained in the interior of the hypercube from those
that do not intersect it. To ensure that the separators are balanced, to each
separator we add all objects intersecting $H_{n^{1/d}}$ with diameter $\geq
1/4$. Note that these objects can be stabbed with $O(1)$ points, and therefore
do not contribute too many cliques to the partition.

\begin{lemma}
Each separator $H_i$ is balanced, in the sense that both the interior and exterior contain at most $\frac{6^d}{6^d+1}n$ objects.
\end{lemma}

\begin{proof}
Due to the choice of $H_0$ the interior of each separator contains at least
$n/(6^d + 1)$ objects. Thus, the exterior of each separator contains at most
the claimed number of objects.

To see that the interior does not contain too many objects either, consider
the separator associated with $H_{n^{1/d}}$. Since all objects with radius
$\geq 1/4$ are contained in the separator, we only need to show that
$H_{n^{1/d}}$ contains at most $(n 6^d)/(6^d + 1)$ objects with radius $<
1/4$. Note that $H_{n^{1/d}}$ has side length $3$, and thus volume $3^d$.
Consider a subdivision of $H_{n^{1/d}}$ into $6^d$ sub-hypercubes of side
length $1/2$. The objects (of radius at most $1/4$) intersecting any such
given sub-hypercube are contained in a hypercube of edge length strictly less
than one. Note that $H_0$ is the smallest hypercube containing at least
$n/(6^d + 1)$ objects. Therefore, at most $n/(6^d + 1)$ objects intersect each
sub-hypercube, and thus $H_{n^{1/d}}$ contains at most $(n 6^d)/(6^d + 1)$
objects of radius $< 1/4$.
\end{proof}

Next, we show that among the separators $H_1,\ldots, H_{n^{1/d}}$, at least
one has a suitable partition into cliques. Consider a separator $S$ and a
partition of $S$ into cliques $\mathcal{C}(S)=C_1,\ldots,C_k$. Then the weight
of $S$ is $\Sigma_{C\in \mathcal{C}(S)} \gamma(|C|)$, where $\gamma$ is a
weight function and $|C|$ denotes the number of vertices of the clique $C$.
We set
$\gamma(n)=\log(n+1)$, but the result holds for any function
$\gamma(n)=O(n^{1-1/d})$.

Given a partition of $S$ into cliques, the number of independent sets in $S$
is at most
\[\prod_{C\in \mathcal{C}(S)} (|C| + 1) = 2^{\sum_{C\in \mathcal{C}(S)}
\log(|C| + 1)}.\]

We show that the total weight of all separators is
$O(n\alpha)$; since there are $n^{1-1/d}$ candidate separators, it follows that
there exists a separator with weight $O(n^{1-1/d}\alpha)$. Such a separator
therefore has $2^{O(n^{1-1/d}\alpha)}$ independent sets.

In the following, let $\beta$ denote the volume of the circumscribed ball of
the smallest object, and note that since the objects are similarly sized, all
circumscribed balls have the same volume up to a constant factor. Note that,
because $H_0$ contains $n/6^d$ objects and we performed a scaling such that
$H_0$ has size $1$, we have $\beta<1$.  We distinguish two cases: if
$\beta^{1/d} > {n^{-1/d}}$ or $\beta^{1/d} \leq {n^{-1/d}}$.

\begin{description}
\item[Case 1:] $\beta^{1/d} > {n^{-1/d}}$. \\

Since $\beta<1$, all balls intersecting the separator are contained in a
hypercube $O(1+\beta) = O(1)$, which can be covered by $O(1/\beta)$ balls of
volume $\beta$. By the definition of the stabbing number, it is possible to
stab all the objects intersecting the separators using $O(\frac{1}{\beta}
\alpha^d)$ points, and thus there is a partition of the objects into
$O(\frac{1}{\beta} \alpha^d)$ cliques, which we denote by $C_1,\dots,C_k$.

The total weight of the cliques $C_1,\dots,C_k$ is

\[\sum_{i=1}^k\gamma(|C_i|)=O(\sum_{i=1}^k (|C_i|)^{1-1/d}).\]

The right hand side here is maximized if the number of cliques is maximum
(i.e., we have $c\frac{1}{\beta} \alpha^d$ cliques for some constant $c$) and each
clique contains the same number of objects (i.e., $\frac{n}{c\frac{1}{\beta}
\alpha^d}$ objects). Furthermore, since the diameter of the union of objects in
any clique is $O(\beta^{1/d})$ and the distance between consecutive separators
is $1/n^{1/d}$, each clique contributes weight to at most
$O(\beta^{1/d}n^{1/d})$ separators. Therefore, the total weight (of all
separators) is at most
\[O\left(\beta^{1/d}n^{1/d} \cdot \frac{1}{\beta}
\alpha^d \cdot \gamma\left(\frac{n}{\frac{1}{\beta} \alpha^d}\right)\right) 
= O(n\alpha),\]
since $\gamma(t)=O(1^{1-1/d})$.
There are $n^{1/d}$ separators, thus at least one of them must have
weight at most $O(n^{1-1/d}\alpha)$.

\item[Case 2:] $\beta^{1/d} \leq {n^{-1/d}}$.\\
Each clique contributes to
the weight of at most $O(1)$ separators. The total weight of all separators is
then at most a constant times the total weight of the cliques. This can be
upper bounded by $O(\sum_{i=1}^k |C_i|^{1-1/d})$, which by the concavity of
$x^{1-1/d}$ is $O(n)$. Thus, there is a separator with weight at most
$O(n^{1-1/d})$.\qedhere
\end{description}
\end{proof}

\subsection{A parameterized algorithm with a sphere separator}

To prove Theorem~\ref{thm:paramalg}, we use the following separator theorem, due to Miller et al. \cite{MillerTTV97}. The \emph{ply} of a set of objects in $\Reals^d$ is the largest number $p$ such that there exists a point $x\in \Reals^d$ which is contained in $p$ objects. 

\begin{theorem}[Miller et al. \cite{MillerTTV97}]\label{thm:millersep}
Let $\Gamma=\{B_1,\ldots,B_n\}$ be a collection of $n$ closed balls in $\Reals^d$ with ply at most $p$. Then there exists a sphere $S$ whose boundary intersects at most $O(p^{1/d} n^{1-1/d})$ balls, and the number of balls in $\Gamma$ disjoint from $S$ that fall inside and outside $S$ are both at most $\frac{d+1}{d+2}n$.
\end{theorem}

We can now prove Theorem~\ref{thm:paramalg}.

\begin{proof}[Proof of Theorem~\ref{thm:paramalg}]
Let $F$ be the set of similarly-sized objects with stabbing number $\alpha$
defining the intersection graph.
Consider the set of balls $B$ made up by the circumscribed balls of the objects of $F$ that are in a
maximum independent set.  We claim that the ply of this set is $O(\alpha^d)$.
To prove the claim, let $S$ be a subset of the independent set whose circumscribed balls overlap
at a point $x\in \Reals^d$. Since the objects are similarly sized, $S$ must
lie within a ball centered at $x$ whose radius is at
most a constant times the diameter of the largest object. Thus, $S$ can be
stabbed by $O(\alpha^d)$ points. However, as $S$ forms an independent set,
each point can only stab at most one object from $S$. Therefore,
$|S|=O(\alpha^d)$.

By Theorem~\ref{thm:millersep} the ball set $B$ has a
$\frac{d+1}{d+2}$-balanced sphere separator, where the sphere intersects $O((\alpha^d)^{1/d}
k^{1-1/d})$ $= O(k^{1-1/d}\alpha)$ balls. We proceed by guessing such a sphere,
but in order to do that, we need to define a polynomially large set of spheres
to guess from.

All that is important about a sphere $\sig$ is the separation that it performs
on $B$, that is, it splits $B$ to the set of balls inside, the set of balls
outside, and the set of balls intersected by $\sig$. Given an arbitrary sphere $\sig$, we shrink it while we can without making it disjoint from any of the originally intersected balls, or until a new ball is touched that was inside the sphere originally. As a result, we get a \emph{canonical} sphere $\sig'$ that is tangent to some set of balls from $B$. Note that such spheres can be uniquely defined by a set of at most $d+1$ tangent balls, and a string that for each of these balls describes if they are inside or outside $\sig'$. In order to define $\sig$, we add another bit for each touching ball, which is set if and only if the ball was originally not intersected by $\sig$. Therefore, the number of guesses we can make for $\sig$ is $n^{d+1}4^{d+1}$. Notice that the guess defines the sets of balls
inside, outside and intersected by $\sig$ as well.

After guessing $\sig$, we proceed by guessing which of the objects intersected by $\sig$ are
in the solution, and remove the remaining objects intersected by $\sig$.
Since at most $O(k^{1-1/d}\alpha)$ of the intersected objects are in the
solution, there are $n^{O(k^{1-1/d}\alpha)}$ possibilities for this guess.

From the remaining objects, we remove those that are adjacent to the objects
guessed to be in the solution, and recurse on the objects inside 
$\sig$ and on the objects outside $\sig$ separately. The running time
$T(n,k)$ for this algorithm satisfies the recurrence (for fixed $d$):

\[T(n,k) = n^{O(k^{1-1/d}\alpha)} \cdot T\left(n, k \cdot \frac{d+1}{d+2}\right)\]

which implies the running time $T(n,k)=n^{O(k^{1-1/d}\alpha)}$. 
\end{proof}


For arbitrary size objects that are $O(1)$-fat in some stronger sense (or just $O(1)$-stabbed), we can apply the above scheme of guessing a separating sphere or hypercube, and use one of the many separator theorems designed for objects of small ply. See~\cite{SmithW98,Chan03,Har-PeledQ17}. One can also apply~\cite{BergBKMZ18} since in case of ply $1$, the weights are constants; although the theorem is stated for the usual notion of fatness, the proof itself uses only the stabbing number. We get the following theorem.

\begin{theorem}
Let $2\leq d \in \Nats$. There is a parameterized algorithm that solves
\IS for intersection graphs of $O(1)$-stabbed
objects in $\Reals^d$ running in time $n^{O(k^{1-1/d})}$, where the parameter
$k$ is the size of the maximum independent set.
\end{theorem}

\section{Wiring in a blowup of the Euclidean Cube}\label{sec:wire}

Our wiring theorem relies on the folklore observation that can be informally stated the following way: an $n \times m$ matrix can be sorted by first permuting the elements within each row, then permuting the elements within each column, and then permuting the elements in each row again. Note that the permutations are independent of each other, and they are not sorting steps; the permutations required are quite specialized. We state the lemma in a more group-theoretic setting. Let $Sym(X)$ denote the symmetric group on the set $X$.

\begin{lemma}[Lemma 4 of~\cite{abert2002symmetric}]\label{lem:symgroup}
Let $A$ and $B$ be two finite sets. Then $Sym(A \times B) = G_AG_BG_A$, where
$G_A$ is the subgroup of $Sym(A \times B)$ consisting of permutations $\pi$
where $\pi(a,b) \in A\times \{b\}$ for all $(a,b)\in A\times B$, and $G_B$ is
the subgroup of $Sym(A \times B)$ consisting of permutations $\pi$ where
$\pi(a,b) \in \{a\}\times B$ for all $(a,b)\in A\times B$.
\end{lemma}

\begin{corollary}\label{cor:symgroups}
Let $2 \le d \in \Nats$ and let $A_1,A_2,\dots,A_d$ be finite sets. Then $\Gamma \eqdef Sym(A_1 \times A_2 \times \dots \times A_d)$ is of the form $\Gamma=G_1G_2\dots G_{d-1}G_dG_{d-1}G_{d-2}\dots G_1$, where $G_i$ is the subgroup of $\Gamma$ consisting of permutations $\pi$ where $\pi(a_1,\dots,a_i,\dots,a_d) \in \{a_1\}\times \dots \times \{a_{i-1}\} \times A_i \times \{a_{i+1}\} \times \dots \times\{a_d\} $ for all $(a_1,\dots,a_d)\in \Gamma$.
\end{corollary}

\begin{proof}
We use induction on $d$; for $d=2$, the statement is equivalent to Lemma~\ref{lem:symgroup}. Let $d\ge 3$. We can write $\Gamma$ as $Sym\big((A_1\times \dots \times A_{d-1})\times A_d\big)$, so by induction (for $d=2$), we have that $\Gamma=G_1 \times G_{A_2\times \dots \times A_{d}} \times G_1$. By induction, we also have that $G_{A_2\times \dots \times A_{d}}=G_2\dots G_{d-1}G_dG_{d-1}G_{d-2}\dots G_2$, therefore $\Gamma=G_1G_2\dots G_{d-1}G_dG_{d-1}G_{d-2}\dots G_1$.
\end{proof}

For an integer $n$, let $[n]=\{1,\dots, n\}$. Let $\EC^d(n)$ be the $d$-dimensional Euclidean grid graph whose vertices are $[n]^d$, and
$x,y\in V(G)$ are connected if and only if they are at distance $1$ in $\Reals^d$. For $x \in \Ints^d$ and $S \subset \Ints^d$, we use the shorthand $x+ S \eqdef \{ x+ y \;|\; y \in S\}$.
Let $\BEC^d(n,t)$ denote the $t$-fold blowup of $\EC^d(n)$, where all vertices of $\EC^d(n)$ are exchanged with a clique of size $t$, and vertices in neighboring cliques are connected. More precisely,
\begin{align*}
V(\BEC^d(n,t)) &= [n]^d\times [t]\\
E(\BEC^d(n,t)) &= \big\{  (x,i)(y,j) \;\big|\; x = y \vee (x,y) \in E(\EC^d(n))\big\}.
\end{align*}

Our second key ingredient is the Euclidean Cube Wiring theorem.

\begin{theorem}[Theorem 21 in~\cite{BergBKMZ18}]\label{thm:euclidwire}
Let $3 \le d\in \Ints$. There exists a constant $c$ dependent only on the dimension such that any matching $M$ between $P =[n]^{d-1}\times \{1\}$ and $Q\eqdef [n]^{d-1}\times\{cn\}$ can be embedded in $\EC^d(cn)$, that is, there is a set of vertex disjoint paths connecting $p$ and $q$ in $\EC^d(cn)$ for all $pq \in M$.
\end{theorem}

\begin{theorem}[Blown-up Cube Wiring]\label{thm:becwire}
Let $3 \le d\in \Ints$, and let $n,t$ be positive integers. We consider two opposing facets of the blown-up cube $\cC \eqdef \BEC^d(cn,t)$ (where $c\in \Ints_+$ depends only on $d$):
\begin{align*}
P &\eqdef \big([n]^{d-1} \times \{1\}\big) \times [t]\\
Q &\eqdef \big([n]^{d-1} \times \{cn\}\big) \times [t]
\end{align*}
Any matching $M$ between $P$ and $Q$ can be embedded in $\cC$, that is, there is a constant integer $c$ dependent only on the dimension $d$ such that for any matching $M$ there is a set of vertex disjoint paths connecting $p$ and $q$ in $\BEC^d(cn,t)$ for all $pq \in M$.
\end{theorem}

\begin{proof}
Without loss of generality, suppose that $M$ is a perfect matching between $P$
and $Q$ (this can be ensured by adding dummy edges to $M$ if necessary). Let
$c=c'+2$ where $c'$ is a constant such that cube wiring can be done in height $h=c'n$. Let
$A=[n]^{d-1}$ and let $B=[t]$. The matching $M$ can be regarded as a
permutation $\pi$ of $A\times B$, where $\pi(a,b)=(a',b')$ if
$\big((a,b)(a',b')\big) \in M$.

By Lemma~\ref{lem:symgroup}, there exists a permutation $\pi_A \in G_A$ and $\pi_B,\pi'_B \in G_B$ such that $\pi=\pi'_B\pi_A\pi_B$, where $G_A$ and $G_B$ are defined as in Lemma~\ref{lem:symgroup}. We can think of both $\pi_B$ and $\pi'_B$ as the union of
$n^{d-1}$ distinct permutations of $[t]$. We can realize $\pi_B$ using one
matching: for all $(x,i) \in A\times B$, we add the edge $((x,1),i)((x,2),j)$ to $M_B$,
where $\pi_B(x,i)=(x,j)$. As a result, $M_B$ is a perfect matching between $P$
and the next layer of the blown-up cube, $P'\eqdef \big([n]^{d-1} \times
\{2\}\big) \times [t]$. Similarly, for all $(x,i) \in A\times B$, let
$M'_B$ contain the edge $((x,cn-1),i)((x,cn),j)$, where $\pi'_B(x,i)=(x,j)$; this matches $Q'
\eqdef \big([n]^{d-1} \times \{cn-1\}\big) \times [t]$ to $Q$. Finally, by
the Cube Wiring Theorem (Theorem~\ref{thm:euclidwire}), for
each $i\in [t]$, there are vertex disjoint
paths from $P'_i \eqdef \big([n]^{d-1} \times \{2\}\big) \times \{i\}$ to
$Q'_i \eqdef \big([n]^{d-1} \times \{cn-1\}\big) \times \{i\} $ that realizes
the matching
\[M^i_A \eqdef \{ (x,i)(y,i) \mid x\in [n]^{d-1}\text{ and } \pi_A(x,i)=(y,i)\}.\]
For each $i\in [t]$, these wirings are vertex disjoint since they are contained
in vertex disjoint Euclidean grid hypercubes. The
matchings $M^i_A$ for $i \in [t]$ together with the matchings
$M_B$ and $M'_B$ realize the matching $M$.
\end{proof}

\section{Lower bounds for packing isometric axis-parallel boxes}\label{sec:lower}

Our first lower bound shows that the running time of the algorithm in
Theorem~\ref{thm:weightedalg} is tight under ETH.

\emph{Overview of the proof of Theorem~\ref{thm:generallower}.}
Our proof is a reduction from $(3,3)$-\textsc{SAT}, the satisfiability problem
of CNF formulas where clauses have size at most three and each variable occurs
at most three times. Such formulas have the property that if they have $n$
variables, then they have $O(n)$ clauses. The problem has no $2^{o(n)}$
algorithm under ETH~\cite{frameworkpaper}.

The proof has two
steps; the first step is a reduction form $(3,3)$-\textsc{SAT} to \IS in
certain subgraphs of the blown-up Euclidean cube, and the second step is to
show that these subgraphs can essentially be realized with axis-parallel boxes.
Throughout the proof, we consider the dimension $d$ to be a constant.

The \emph{incidence graph} of a $(3,3)$-CNF formula $\phi$ is a graph where vertices
correspond to clauses and variables of $\phi$, and a variable and clause
vertex are connected if and only if the variable occurs in the clause.

\subsection{\IS in subgraphs of the blown-up Euclidean cube}
\label{subsec:loweroverview}

\paragraph*{A simple and generic lower bound construction for \IS.}

We give a generic reduction from \textsc{$(3,3)$-SAT} to \IS, which serves as
a skeleton for the more geometric type of reduction we will do later.

Consider the incidence graph of $\phi$. Replace each variable vertex $v$ with
a cycle of length $6$, consisting of vertices $v^1,\dots,v^{6}$, where the
edges formerly incident to $v$ are now connected to distinct cycle vertices $v^2,v^4$ or $v^6$ for positive literals and to $v^1,v^3$ or $v^5$ for negative literals (see
Figure~\ref{fig:vcgadget}). We replace each clause vertex $w$ that corresponds to a
clause of exactly $3$ literals with a cycle of length three,
and connect the formerly incident edges to distinct vertices of the triangle.
For clauses that have exactly two literals, the gadget is a single edge, and
we connect the formerly incident edges to distinct endpoints of the
edge. We can eliminate clauses of size $1$ in a preprocessing step. Let
$G'_\phi$ be the resulting graph.

\begin{figure}
\begin{center}
\includegraphics[scale=1]{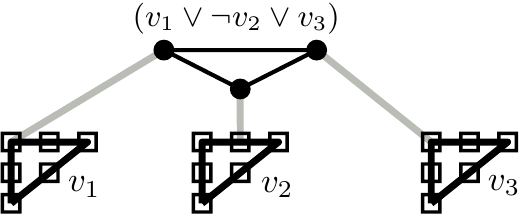}
\caption{The graph $G_\phi$ for
$\phi=(v_1 \vee \neg v_2 \vee v_3)$}\label{fig:vcgadget}
\end{center}
\end{figure}

An independent set can contain at most $3$ vertices of a variable cycle of
length $6$, and at most $1$ vertex per clause gadget. Observe that a formula with $\nu$ variables and $\gamma$
clauses has an independent set of size $3\nu + \gamma$ if and only if
the original formula is satisfiable.

Let $G$ be a graph and let $uv$ be an edge of $G$. A \emph{double subdivision}
of $uv$ is replacing $uv$ with a path of length $3$, i.e., we add the new
vertices $w$ and $w'$, remove the edge $uv$ and add the edges
$uw, ww', w'v$. A graph that can be obtained from $G$ by
some sequence of double subdivisions is called an \emph{even subdivision} of
$G$. Observe that a double subdivision increases the size of the maximum
independent set by one, so $G$ has an independent set of size $k$ if and only
if its even subdivision $G'$ has an independent set of size
$k+\frac{|V(G')|-|V(G)|}{2}$.

\vspace*{-0.5em}
\paragraph*{Embedding $G'_\phi$ into a blown-up cube.}

In a blown-up cube $\BEC^{d}(n,t)$, we call a clique corresponding to $x\in
[n]^d$ the \emph{cell of $x$} or simply a cell, that is, the cell of $x$ is
defined as the set of vertices $\{x\} \times [t] \subset V(\BEC^{d}(n,t))$.

The following is a tight lower bound for \IS inside the blown-up
Euclidean cube.

\begin{theorem}\label{thm:indepinbec}
For any fixed constant $d\geq 3$, there exists a $\gamma>0$ such that for any
$t \geq 2$ there is no $2^{\gamma n^{1-1/d}t^{1/d}}$ algorithm
for \IS for subgraphs of the blown-up cube $\cC \eqdef \BEC^{d}((n/t)^{1/d},t)$
under ETH. The
lower bound holds even if the subgraph $G$ has maximum degree three, and the
neighbors of each vertex in $G$ lie in distinct cells.
\end{theorem}

\begin{proof}
Given a $(3,3)$-SAT formula $\phi$, we show that we can construct a subgraph
of a blown-up cube with the required properties that is also an even
subdivision of $G'_\phi$. If $\phi$ has $\bar{n}$ literals, then we create a
subgraph $G$ that has $n=c \cdot \bar{n}^{\frac{d}{d-1}}/t^{\frac{1}{d-1}}$ vertices for some constant $c$ which will be specified later.
Let $s \eqdef (n/t)^{1/d}$ denote the side length of $\cC$. Note that $|V(\cC)|=s^d t =n$.

First, we embed an even subdivision of $G'_{\phi}$ into $\cC$ as explained next. We use the bottom and the top ``layers'' of the blown-up
cube to embed the variable cycles and clause cycles respectively. Let $P$ and $Q$ be
the point sets corresponding to the bottom and top layer of cells respectively, i.e.,
$
P \eqdef \big([s]^{d-1} \times \{1\}\big) \times [t],
\qquad
Q \eqdef \big([s]^{d-1} \times \{s\}\big) \times [t].
$
We embed each variable cycle of $G'_\phi$ into $P$: we (injectively) associate six vertices of even
intra-cell index in six cells of $P$, that is, a variable cycle on vertices
$v_1\dots v_6$ is associated with the vertices
$(x^{(1)},2k),\dots,(x^{(6)},2k)$ in this order, where
$x^{(1)},\dots x^{(6)}$ is a cycle in  the bottom facet of $\EC^d(s)$ and $k
\in [t/2]$. As we will see below, if $|P|$ is large enough, then we can pick $x^{(1)},\dots x^{(6)}$ and $k$ for each variable cycle so that this association is injective.

With each clause, we associate a pair or triplet of vertices in $Q$
that are in neighboring cells, more precisely, for clauses of size three, the
vertices are of the form $(x,k),(x,k+1),(x',k)$, while for clauses of size two
we have $(x,k),(x',k)$ for some $(x,x')\in E(\EC^d(s))$ and $k\in [t-1]$. If $|Q|$ is large enough, then we can pick $x$ and $k$ for each clause so that the association remains injective.

Let $\Var$ be the set of vertices in $\cC$ corresponding to vertices on the variable cycles with a wire connection. Let $\Cla$ be the set of
vertices in $\cC$ assigned to the clauses. Note that the number of literals is
$|\Cla|=|\Var|=\bar{n}$, and we have that
\[
\begin{split}
|P|=|Q|= s^{d-1} t = 
(n/t)^{\frac{d-1}{d}}t=\left(\frac{c\cdot \bar{n}^{\frac{d}{d-1}}/t^{\frac{1}{d-1}}}{t}\right)^{\frac{d-1}{d}} \cdot t
=c^{\frac{d-1}{d}}\bar{n}.
\end{split}
\]
By picking $c\geq 6$, we ensure that there is
enough space to do the above associations injectively for any $d\geq 3$, as we will have $|P|=|Q|>3\bar{n}$.

Let $M$ be the perfect matching between $\Var$ and $\Cla$ given by $\phi$.
By Theorem~\ref{thm:becwire}, there is a wiring from
$\Var$ to $\Cla$ realizing $M$, as long as $c$ is a large enough constant.
Crucially, observe that $|P|=|Q|=\Theta(s^{d-1}t)$
means that $P$ and $Q$ occupy a constant fraction of the vertices in the cells
of the top and bottom facet of $\cC$, so the wiring given by
Theorem~\ref{thm:becwire} is dense in the sense that a constant fraction of
all vertices of $\cC$ is induced by the wiring.

Next, in $Q$, we add an edge or triangle for each pair or triplet of vertices
assigned to a clause. In $P$, for each vertex $(x,2k) \in \Var$ that is the
endpoint of a wire of even length, we add an edge $((x,2k),(x,2k-1))$, and
regard $(x,2k-1))$ as the new endpoint of this wire. Finally, for each
six-tuple of wire endpoints corresponding to a variable, we add a $6$-cycle.

The graph $G'$ created this way clearly has the desired properties: it is a
subgraph of $\cC$ that has maximum degree three, and the
neighbors of each vertex in $G'$ lie in distinct cells. Moreover, $G'$ can be
constructed in $O(n)=\poly(\bar{n})$ time.

By the properties of even subdivisions of $G'_\phi$, we know that $G'$ has an
independent set of a certain size if and only if $\phi$ is satisfiable.
Suppose that for all $\gamma>0$ there is an
$\exp\left(\gamma n^{1-1/d}t^{1/d}\right)$ algorithm for \IS.
This would result for all $\gamma>0$ in a $(3,3)$-\textsc{SAT} algorithm with running time
\[
\exp\left(\gamma c \cdot \bar{n}^{\frac{d}{d-1}}/t^{\frac{1}{d-1}} 
 \right)^{1-1/d} \cdot t^{1/d} + \poly(\bar{n}) = 2^{(\gamma c)^{1-1/d} \cdot \bar{n}} + \poly(\bar{n}).
\]
The existence of such algorithms contradicts ETH.
\end{proof}


\subsection{Detailed construction and gadgetry}

Having established our lower bound for blown-up Euclidean cubes, we now need
to construct a set of canonical boxes whose intersection graph is an even
subdivision of a given subgraph with maximum degree three where the neighbors
of each vertex lie in distinct cells.

\begin{theorem}\label{thm:construct}
Let $d\geq 3$ and $L\geq 16$ be fixed, and let $G$ be a subgraph of the blown-up
cube $\cC = \BEC^{d}(s,(L/8)^{d-1})$ of maximum degree three, where the
neighbors of each vertex lie in distinct cells. Then $G$ has an even
subdivision $G'$ that can be realized using boxes of size $1\times \dots \times 1 \times
L$. Moreover, given $G$, the boxes of $G'$ can be constructed in $O(|V(\cC)|)$
time, and $|V(G')|=O(|V(G)|)$.
\end{theorem}

We proceed with the proof of Theorem~\ref{thm:construct}.
We consider $d=3$ first; later on, we show how the construction can be generalized to higher
dimensions. We need to define a set of boxes whose intersection graph is an
even subdivision of $G$. The idea is to create a generic \emph{module} that is
able to represent a subgraph of $G$ induced by any cell; these modules will
take up $O(L)\times O(L)\times O(L)$ space. The modules are arranged into a
larger cube of side length $O(sL)$ to make up the final construction.

\begin{figure}[t]
\begin{center}
\includegraphics[height=3cm]{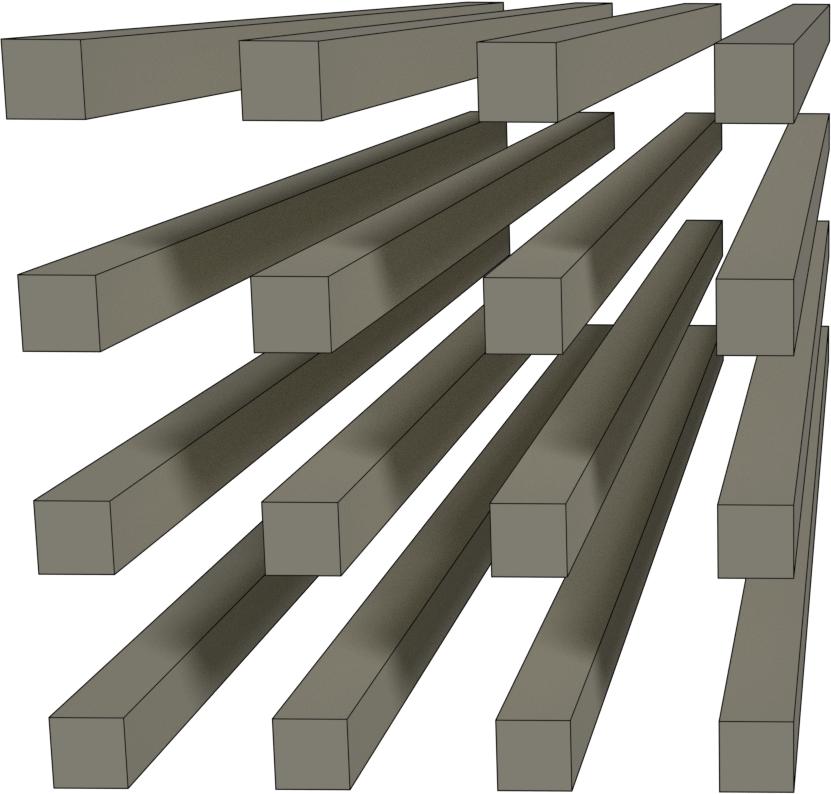}
\caption{A basic brick.}\label{fig:brick}
\end{center}
\end{figure}

\paragraph*{Modules and bricks.}
We index the vertices in a cell by a pair from $[L/8]^2$. 
The starting object in our reduction is a set of $(L/8)^2$ disjoint boxes
parallel to the same axis, arranged loosely in an $L/8 \times L/8$ grid
structure called a \emph{brick}. See Figure~\ref{fig:brick} for an example.
Loosely speaking, each box of each brick within the cell's module can be associated with a vertex of the cell; for a brick $B$, we can
refer to a box corresponding to vertex $(i,j)$ of the cell as $B(i,j)$.

Let $X$ be the set of cells within $\cC$: $X\eqdef \big\{\{x\}\times [L/8]^2
\;\big|\;x \in [cn]^d\big\}$. The wiring within each cell $x\in X$
will be represented by $O(1)$ bricks, and these bricks will fit in an $O(L)$
side length module.

The position of a brick can be specified by defining its axis (along which the
side length of the boxes is  $L$), and for each box $(i,j)$ within the brick,
defining the coordinates of its lexicographically smallest corner (or
\emph{lexmin corner} for short). For example, consider the brick $B$ with axis
$x_3$ where box $B(i,j)$ has coordinates $(3i,3j,0)$. (See Figure~\ref{fig:brick}.) This brick and all
bricks isometric to this are called \emph{basic brick}s. Most bricks can be
thought of as a perturbation of a basic brick, where we apply shifts to each box. 
The eventual module that we create will consist of several bricks, which
together will represent an even subdivision of the sparse graph $G$ restricted to a given
cell. Note that no single brick can be said to represent the set of vertices in a cell.
When defining our gadgetry, it is convenient to talk about these bricks, even
though in the final construction we only need a
certain subset of the boxes within each brick. We can remove the unwanted
boxes from each brick at a later stage.

\paragraph*{Parity Fix, adjustment, bridge, and elbow gadgets.}

The parity fix gadget is introduced so that we can ensure that each of the
subdivisions that we create are even subdivisions. The gadget induces a path of length $3$
or $4$ depending on our needs, but occupies the same space in both cases. More precisely,
the parity fix gadget contains three or four boxes, depending on the parity we
need. The union of the boxes is a larger box of size $3L \times 1 \times 1$;
it is easy to see that within that space we can realize both a path of length
three and four using $L\times 1 \times 1$ boxes: one can cover the larger box
by placing their lexmin corners at equal length intervals.

We can bridge distance along the axis of a basic brick by putting basic bricks
next to each other, where each box intersects only the box of the same index from the
previous and following brick. This creates a set of $(L/8)^2$ vertex disjoint
paths in the intersection graph. We call this a \emph{bridge gadget}.

Using two bricks of the same axis, we can in one step
get rid of a perturbation (or introduce one). Let $B$ be a normal brick with
axis $x_3$ that is a perturbation of the basic brick. We introduce the
basic brick $B'$ that is the translate of the basic brick with the vector
$(0,1,L/2)$. Notice that box $B(i,j)$ intersects $B'(i,j)$ and no other boxes.
Moreover, we could even introduce arbitrary perturbations along the $x_1$
axis in $B'$ and along the $x_3$ axis within both $B$ and $B'$ without
changing the intersection graph induced by $B$ and $B'$. We call a pair of
normal bricks that are a translated and rotated version of these an
\emph{adjustment gadget}.

\begin{figure}[t]
\begin{center}
\includegraphics[height=5cm]{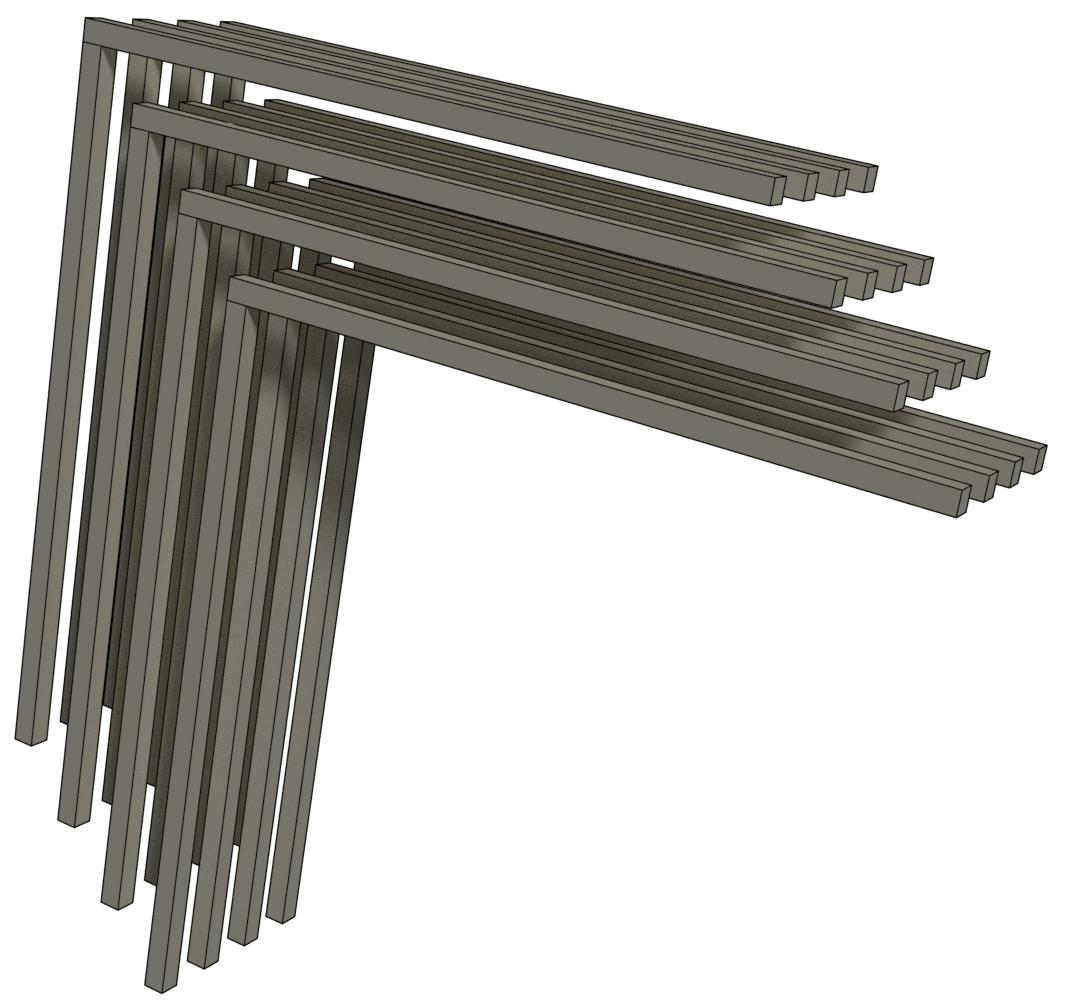}
\caption{An elbow.}\label{fig:brickelbow}
\end{center}
\end{figure}

\begin{figure}[t]
\begin{center}
\includegraphics[height=7cm]{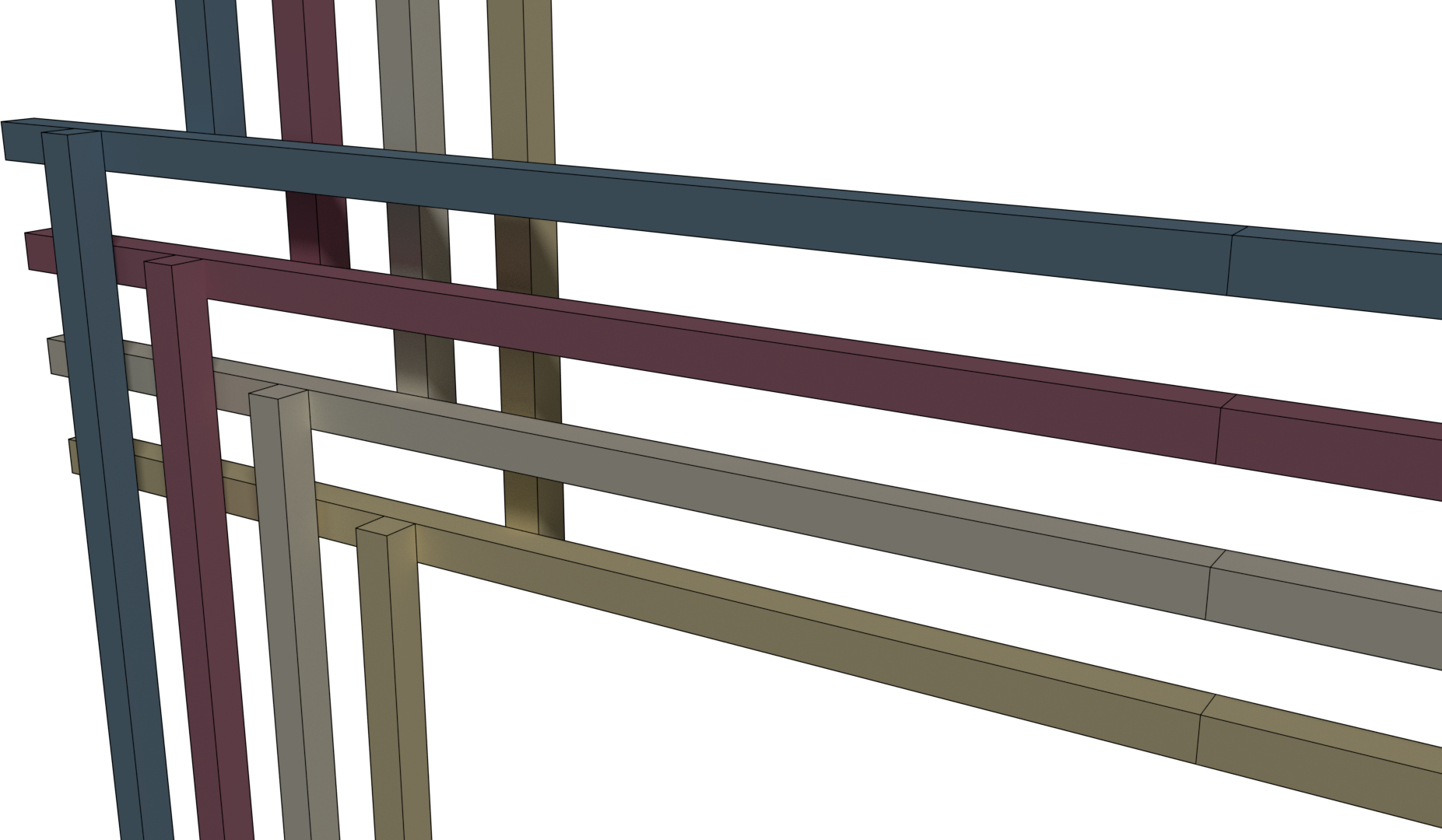}
\caption{The first ``column'' of a branching gadget.}\label{fig:branching}
\end{center}
\end{figure}

Next, we introduce a way to change brick axis using an ``elbow''. Consider a
brick $B$ that is a perturbation of the basic brick, where box $(i,j)$ has
coordinates $(3i,3j,-3i)$. The brick $B'$ has axis $x_1$ and the coordinates
for $B'(i,j)$ are $(3i,3j,L-3i)$ (see Figure~\ref{fig:brickelbow}). Notice
that using these \emph{elbow gadgets} and adjustment gadgets together, one can
route from any brick to any other brick at distance $\Theta(L)$ in $O(1)$
steps.

\paragraph*{The parallel matching gadget.}

A parallel matching
gadget is capable of realizing a matching between two cells where the endpoints of each matching edge
differ only on a fixed coordinate, so for $d=3$, all edges are of the type
$\big((x,(i,j)),(x',(i',j)) \big)$ or all edges are of the type $\big((x,(i,j)),(x',(i,j')) \big)$ for some cells $x$ and $x'$.
We call a matching with this property a \emph{parallel matching}.
Parallel matchings can be decomposed into matchings on disjoint cliques, where
each clique contains vertices that share all of their coordinates except one. In the remainder of this gadget's description, we will omit the cells $x$ and $x'$ from the coordinate lists.

Suppose that each matching edge is of the form $\big((i,j),(i',j)\big)$.
Let $\pi_j(i)$ denote the first coordinate of the pair of
$(i,j)$, that is, suppose that the matching edges are
$\big((i,j),(\pi_j(i),j)\big), \; i \in I_j$ for some sets $I_j \subseteq
[L/8]$. Instead of realizing these matchings, we first extend them to
permutations $\pi_j$ on each clique $[L/8]\times \{j\}$. A permutation can be
thought of as a perfect matching between two copies of a set; by removing the
unwanted vertices (removing the unwanted boxes) we can get to a representation
of the matching, i.e., a set of vertex disjoint paths that connect box $(i,j)$
in the starting brick to box $(\pi_j(i),j)$ in the target brick.


In every brick, each box is translated individually, where the translation
vector's component along the brick's axis  must be an integer $k \in
3\cdot\{-L/8,\dots,L/8\}$, and along the other axes it must be of the form
$k/L$ for some $k \in \{-L/8,\dots,L/8\}$. For a brick $B$, its box of index
$(i,j)$ is denoted by $B(i,j)$, and recall that the position of a box is
defined by its lexmin corner and the axis of the brick.

\begin{figure}[t]
\begin{center}
\includegraphics[height=5.4cm]{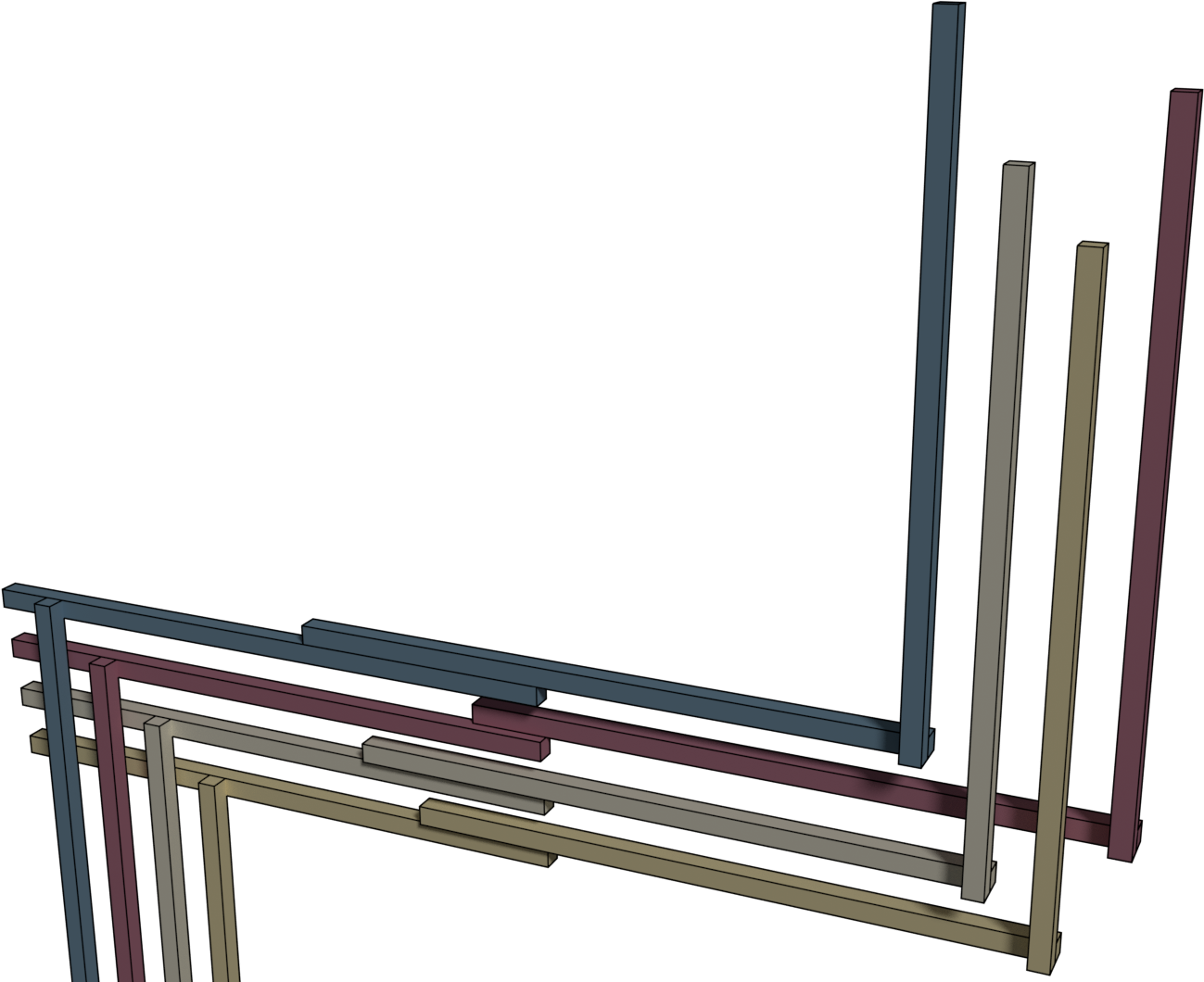}\hfill
\includegraphics[height=5.4cm]{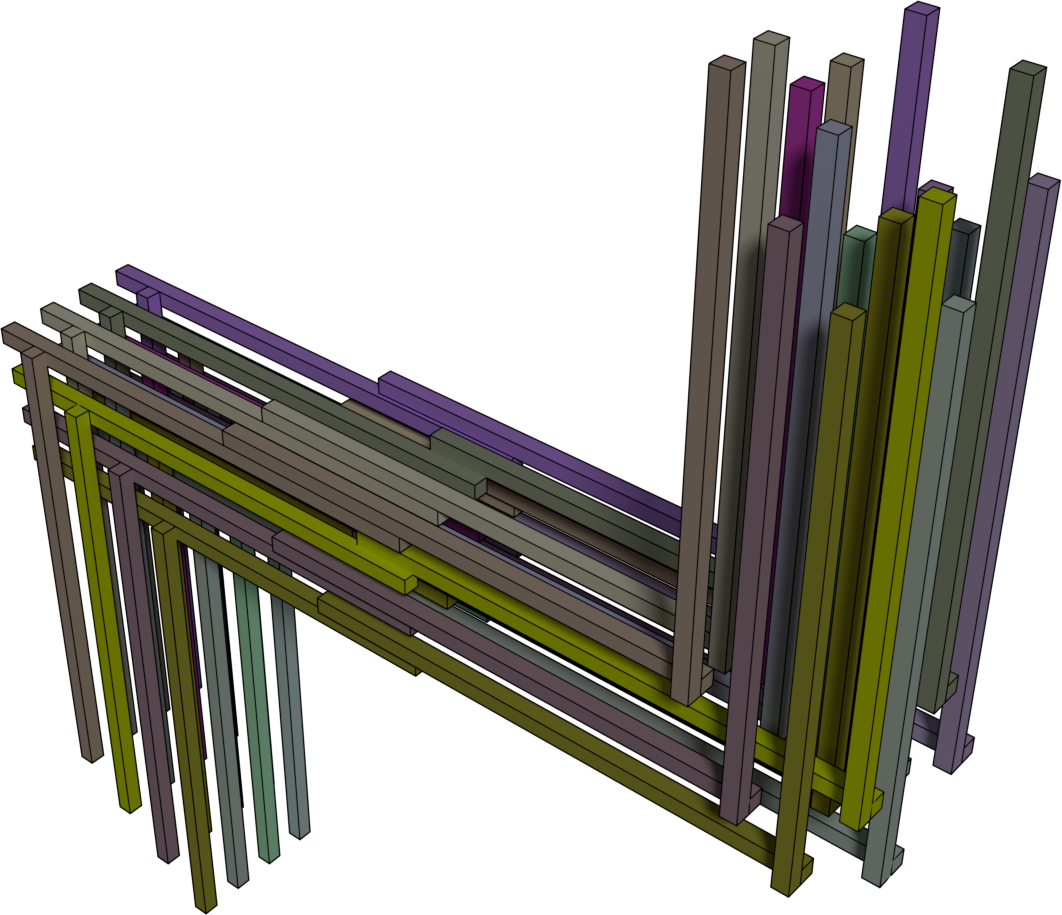}
\caption{Left: First column of a parallel matching gadget for the permutation
$\pi_1(1)=1,\pi_1(2)=4,\pi_1(3)=2,\pi_1(4)=3$. Boxes of each color induce paths; boxes of different color are disjoint. Right: A full parallel matching gadget.}\label{fig:matchingrender}
\end{center}
\end{figure}

We give the coordinates of each box in each brick of the parallel matching gadget below. Let us take the matching edges where $j=1$ first. We start with the first
column of the brick ($j=1$), where the coordinates of $B^{(1)}(i,1)$ are
$(3i,3+i/L,-3i)$.  See the left hand side of
Figure \ref{fig:matchingrender} that illustrates
the idea behind the gadget.
The coordinates for $B^{(1)}(i,j)$ are $(3i,3j+i/L,-3i)$.
The first column of the next brick $B^{(2)}$ has axis $x_1$ and the
coordinates of $B^{(2)}(i,1)$ are $(0,4+i/L,L-1-3i)$, that is, these
boxes touch the previously defined boxes of $B^{(1)}$ from ``behind'' in
Figure~\ref{fig:matchingrender}. In general, $B^{(2)}(i,j)$ has
coordinates $(0,3j+1+i/L,L-3i)$. The next brick $B^{(3)}$ also has axis $x_1$, and the
coordinates for $B^{(3)}(i,j)$ are $(L/2+3\pi_j(i),3j+1+\pi_j(i)/L,L-3i)$,
that is, we change the box perturbations along the first and second
coordinate. Finally, the last brick $B^{(4)}$ has axis $x_3$ and the
coordinates are $(3L/2+3\pi_j(i),3j-\pi_j(i)/L,L-3i)$, i.e., they are placed
``in front of'' the bricks of $B^{(3)}$ in Figure~\ref{fig:matchingrender}.
This can be rewritten as $B^{(4)}(i',j)$ having coordinates
$(3L/2+3i',3j-i'/L,L-3\pi^{-1}_j(i'))$. Notice that in the final
brick, we indeed have the desired ordering, i.e., the ordering of the boxes
along the $x_1$ axis is as required. It is routine to check that the
intersection graph induced each column of this parallel matching gadget
consists of vertex disjoint paths of length four. Different columns are also
disjoint since projecting the boxes of column $j$ onto the $x_2$ axis results
in a subset of the open interval $\oiv{3j-0.5, 3j+2.5}$.

\paragraph*{Realizing an arbitrary matching of a biclique or clique.}

We can regard a general matching $M$ induced by two neighboring cells as a
permutation of $[L/8]^2$, which can be written as the product of three special
permutations by Corollary~\ref{cor:symgroups} that correspond to parallel
matchings; i.e., the matching $M$ is realizable as the succession of three parallel matchings.
This means that each edge of $M$ becomes a path of length three, so by using
three parallel matching gadgets in succession we can represent $M$. We add a
parity fix gadget to each box at the beginning of each wire, which will be
useful later to ensure that each edge has been subdivided an even number of
times.  As a result, we have realized $M$ using $O(1)$ bricks and $O(L)\times
O(L) \times O(L)$ space. This collection of boxes is called a \emph{general
matching} gadget. A general matching gadget has a first and a last brick where
it connects to the rest of the construction, we call these bricks endbricks.

If the goal is to realize a matching within a cell with vertex set $V_x$, then
we can just create two copies of $V_x$ (denoted by $V'_x$ and $V''_x$), with a
complete bipartite graph between them. For a matching edge $v_iv_j \in
\binom{V_x}{2}$, we identify it with the edge $v_i' v_j''$. Then we realize
the matching of this biclique using a general matching gadget.

\paragraph*{The branching gadget.}

The \emph{branching gadget} creates for all indices in $[L/8]^2$ a disjoint
copy of a star on $4$ vertices (that is, a vertex of degree $3$ with its
neighborhood of $3$ isolated vertices). This gadget contains four bricks, and
realizes $(L/8)^2$ disjoint stars. We use the first two bricks ($B^{(1)}$
and $B^{(2)}$) of the parallel matching gadget. The third brick $B'$ is a
translate of the first brick $B^{(1)}$ with the vector $(3,2,L-1)$, i.e., the
coordinates of $B'(i,j)$ are $(3i+3,3j+2+i/L,L-1-3i)$. The final brick $B''$ is
the translate of $B^{(2)}$ by the vector $(L,0,0)$.  See
Figure~\ref{fig:branching} for a rendering of the first ``column'' of the four
bricks. Vertices corresponding to $B^{(2)}$ have degree three, and their
neighbors are the boxes of the same index in $B^{(1)}$, $B'$ and $B''$.

\paragraph*{Constructing a module.}

Our goal is to define modules of side length $O(L)$ that are capable of
representing the role played by cells. The modules together must be able to
represent a subgraph of $\cC$ of maximum degree three, where the neighbors of
any vertex lie in distinct cells.

For all pairs of neighboring modules, we introduce a general matching gadget
to represent the matching required by $G$ between the two neighboring cells.
These gadgets form the \emph{interface}. Moreover, in the middle of each
module, we add another general matching gadget to represent the matching
within the cell; this gadget is the \emph{core} of the module. See
Figure~\ref{fig:module}. Finally, within each module, we tie the endbricks of
the core and the endbricks of the interface falling inside the module together
with a \emph{brick-tree}. The brick-tree is a collection of $(L/8)^2$
isomorphic and disjoint trees, realized as a collection of branching, elbow,
adjustment and bridge gadgets. Each tree $(i,j)$ has maximum degree three, and
its leaves are the boxes of index $(i,j)$ in the interface and in the core.

\begin{figure}
\begin{center}
\includegraphics[scale=0.8]{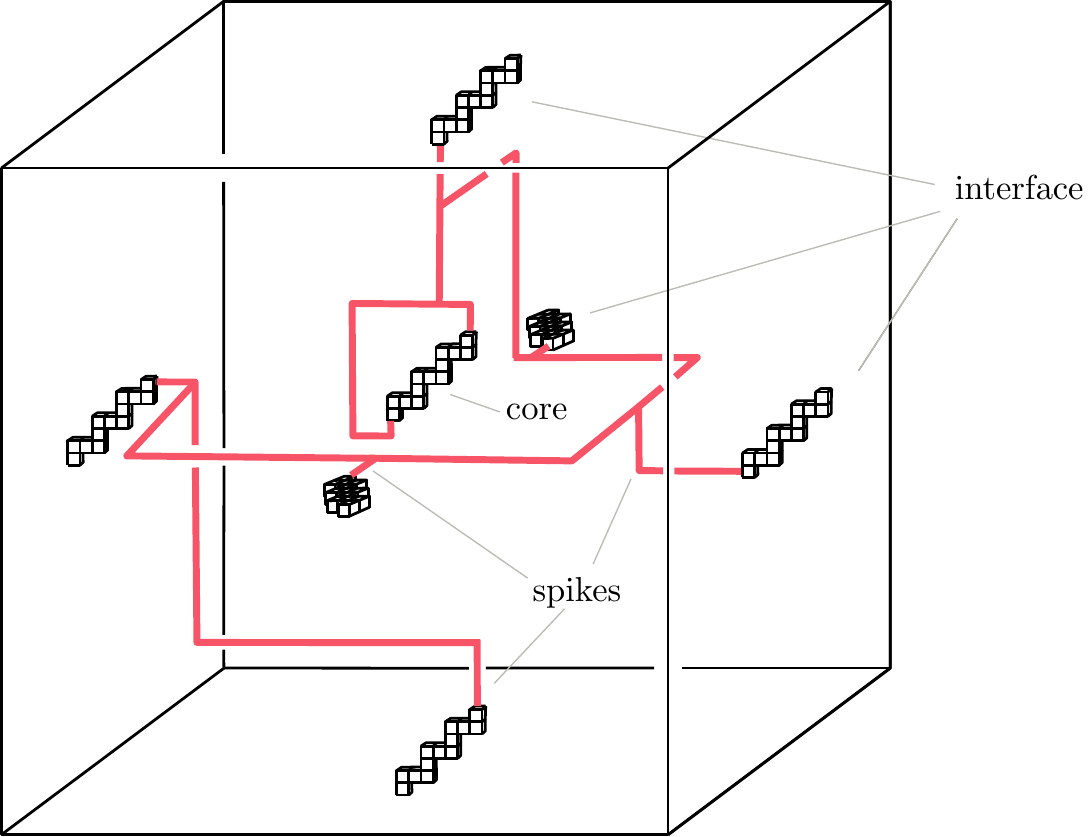}
\caption{A module with general matching gadgets of the interface and the
core, with the simplified image of a brick-tree (in red).}\label{fig:module}
\end{center}
\end{figure}

First, we show that such a construction is sufficient to represent an even
subdivision of an arbitrary subgraph $G$, and later we show how the brick-tree
can be constructed. Let $G$ be a subgraph with the desired properties, and let
$x$ be a particular cell. For each edge $uv$ induced by $x$, we fix an
arbitrary orientation, and realize the acquired matching so that the source
vertex of the arcs are in one end of the core and the targets are in the
other. Since the neighbors of any vertex lie in different cells, all indices
of $[L/8]^2$ appear at most once, either as a source of an arc, as a target of
an arc, or not at all. Then we realize the arcs using the core's general
matching gadget of the module. For each index $\bi\in [L/8]^2$, the edges
incident to vertex $\bi$ of $x$ can be assigned to a subtree $T$ of the tree
corresponding to index $\bi$, where $T$ has at most three leaves, at most one
of which is adjacent to a box of the core, and other leaves are adjacent to
boxes in distinct endbricks of the interface. There is a unique minimal
subtree $T$ that induces the desired (at most three) leaves; we can map a
vertex $v\in V(G)$ of degree three to the degree three vertex of $T$. If $V$
has a smaller degree, then it can be mapped to an arbitrary non-leaf vertex of
$T$.

To construct a brick-tree in $\Reals^3$, consider first a Euclidean grid cube
of size $O(1)$. We can use this small cube as a model of our module: in
general, an edge of this cube represents a brick. We have some edges already
occupied by the general matching gadgets corresponding to the interface and
the core. By choosing a cube large enough, we can ensure that these vertices
are distant in the $\ell_1$ norm. It is easy to see that if the cube is large
enough (we allow its size to depend only on $d$), then there is a subtree of
the grid of maximum degree three, where the leaves are some distant prescribed
vertices. Such a tree can be constructed for example by mimicking a
Hamiltonian path of the inscribed octahedron of the module, and adding to it
small ``spikes'' that go to the endbrick of the interfaces. At the end of the
path, we extend it towards the center of the cube, where we add another
branching for the two endbricks of the core. The branching points in the
brick-tree are branching gadgets, the turns are elbow gadgets, and straight
segments are bridges and adjustments.

\vspace*{-0.5em}
\paragraph*{Finalizing the construction in $\Reals^3$.}

By packing the modules in a side length $O(sL)$ Euclidean cube, and removing
unused boxes from each module according to the given subgraph, we get our
final construction for three dimensions. For each edge, we have it represented
by a sequence of $O(1)$ boxes passing through a single general matching
gadget. Using the parity fix gadget inside the general matching gadget, we can
ensure that the path representing the edge has an odd number of internal
vertices. Therefore, the final construction has $O(|V(G)|)$ boxes, and each
edge of $G$ is represented with a path of odd length, that is, the graph
induced by the boxes is an even subdivision of $G$.

\vspace*{-0.5em}
\paragraph*{The construction in higher dimensions}

It is surprisingly easy to adapt our three-dimensional construction to the
$d$-dimensional case. This time, we need to realize a subgraph of
$\cC=\BEC^{d}(s,(L/8)^{d-1})$.

The basic brick in $d$ dimensions contains $(L/8)^{d-1}$ boxes, indexed by
$[L/8]^{d-1}$, where the lexicographically minimal corner of box $\bi$ is
$(0,3\bi)$. For normal bricks, we allow perturbations of the form $3k\; (|k|
\in [L/8])$ along the axis of the brick, and $k/L\; (|k|\in [L/8])$ in all
other directions. The parity fix, adjustment, and elbow gadgets can be defined
analogously. The parallel matching gadget is also straightforward: the task
here is to represent a parallel matching, where each edge is of the form
$(\bi,\bi')\in [L/8]^{d-1}\times [L/8]^{d-1}$, where $\bi,\bi'$ differ only on
the $t$-th coordinate for some fixed $t\in [d-1]$. As previously, we can
extend this to $(L/8)^{d-2}$ permutations, where for each $\iota\in
[L/8]^{d-2}$, we have a permutation $\pi_{\iota}$ over the ``column'' $\iota$,
i.e., over the set
\[
\lbrace(i_1,\dots,i_{d-1})\;\big|\;i_t\in [L/8]
\text{ and }(i_1,\dots,i_{t-1},i_{t+1},\dots,i_{d-1})=\iota\rbrace.
\]
Such a permutation can be represented as described before: we replace the role
played by the $x_1$ axis with $x_t$, the role of $x_2$ with $x_{t+1 \bmod
(d-1)}$ and $x_3$ with $x_d$. Along all other axes, we introduce no
perturbations to the boxes. The column gadget corresponding to column
$\iota=(i_1,\dots,i_{t-1},i_{t+1},\dots,i_{d-1})$ can be covered
by\footnote{The formula is only accurate for the case $t\leq d-2$. If $t=d-1$,
the role of $x_{t+1}$ and $x_1$ should be switched.}
\[
\begin{split}
[3i_1,3i_1+1]&\times \dots \times [3i_{t-1},3i_{t-1}+1]\\
&\times[-L/2,3L/2]\times (3i_{t+1}-0.5,3i_{t+1}+2.5)]\\
&\times [3i_{t+2},3i_{t+2}+1] \times \dots
\times [3i_{d-1},3i_{d-1}+1] \times [0,\frac{3}{2}L].
\end{split}
\]
These sets are clearly disjoint for distinct values of $\iota$.

A general matching $M$ is regarded as a permutation of $[L/8]^{d-1}$, which
can be written as the product of $2(d-1)-1$ special permutations by
Corollary~\ref{cor:symgroups} that correspond to parallel matchings;
therefore, $M$ is realizable as the succession of $2d-3$ parallel matchings.
As a result, we can realize $M$ with $O(d)=O(1)$ bricks and $O(L)\times \dots
\times O(L)$ space. As before, we add parity fix gadgets to each box of one of
the endbricks.

To realize a brick-tree, we can again trace a Hamiltonian path of the graph
given by the dimension $1$ faces of the cross-polytope inside the module, and
add spikes to it to reach the endbricks of the interface and extend it to the
two endbricks of the core. Note that the cross-polytope does have a
Hamiltonian path, we can use e.g.
\[
(1,0,\dots,0);(0,1,0,\dots,0)\dots (0,\dots,0,1);
(-1,0,\dots,0);(0,-1,0,\dots,0) \dots (0,\dots,0,-1).
\]
The finalizing steps are again analogous to the $3$-dimensional case. This
concludes the proof of Theorem~\ref{thm:construct}.

\medskip

Using Theorem~\ref{thm:construct}, it is easy to
prove Theorem~\ref{thm:generallower}.

\begin{proof}[Proof of Theorem~\ref{thm:generallower}]

Set $L \eqdef \max(16,\alpha^{\frac{d}{d-1}})$. This choice of $L$ implies that
any family of canonical boxes of size $1\times 1 \times L$ are $O(\alpha)$-stabbed.
Furthermore, set $t=(L/8)^{d-1}$. The proof is by
reduction from \IS on subgraphs of the blown-up cube $\cC \eqdef
\BEC^{d}((\bar{n}/t)^{1/d},t)$, where the subgraph $G$ has
maximum degree three, and the neighbors of each vertex in $G$ lie in distinct
cells.  By Theorem~\ref{thm:indepinbec}, there is no $\gamma>0$ for which a
$2^{\gamma n^{1-1/d}t^{1/d}}$ algorithm exists for this problem under
ETH.

Let $G$ be a subgraph of $\cC$ as described above. By
Theorem~\ref{thm:construct}, we can realize an odd subdivision $G'$ of $G$
using boxes of size $1\times \dots\times1\times L$, with $O(\bar{n})$ vertices
in $\poly(\bar{n})$ time. If for any $\gamma>0$
there is an algorithm for \IS on $\alpha$-stabbed canonical boxes with running
time $2^{\gamma n^{1-1/d}\alpha}$, then this translates into $2^{\gamma
n^{1-1/d}L^{1-1/d}}$ algorithms for all $\gamma>0$. This can be composed with
our construction to get  $2^{\gamma \bar{n}^{(1-1/d)}t^{1/d}}$ algorithms for
all $\gamma>0$ for \IS on the described subgraphs of $\cC$, which contradicts
ETH according to Theorem~\ref{thm:indepinbec}.
\end{proof}

\section{A parameterized lower bound:
the proof sketch of Theorem~\ref{thm:paramlower}}\label{sec:app_paramlower}

In this section, show that a construction similar to the previous one yields a
parameterized lower bound as well, which almost matches the parameterized
algorithm given in Theorem~\ref{thm:paramalg}. Due to many similarities, we
only sketch this proof.

Our proof is a direct reduction from the \textsc{Partitioned Subgraph
Isomorphism} problem~\cite{Marx10}, where one is given a graph $G$ whose
vertex set is partitioned into the sets $A_1,A_2,\dots,A_k$, and a $3$-regular
graph $H$ on $k$ vertices $V(H)=\{v_1,\dots, v_k\}$, and the goal is to find a
subgraph isomorphism $\phi:V(H) \to V(G)$ such that $\phi(v_i) \in A_i$ for
all $1\leq i \leq k$. We know that there is no $f(k)n^{\gamma k/\log k}$ algorithm
for any $\gamma>0$ for \textsc{Partitioned Subgraph Isomorphism}~\cite{Marx10},
unless ETH fails.

Therefore, given an instance of \textsc{Partitioned Subgraph Isomorphism}, our
task is to construct a set of axis-parallel boxes that have independent set of
a certain size $g(k)$ if and only if there is a partitioned subgraph
isomorphism from $H$ to $G$. We will use modules that are very similar to the
modules used earlier, arranged in a larger Euclidean hypercube, but this time
we will only use the modules to realize the wires, and we add special gadgets
(equalizer and edge check) outside the modules in the top and bottom facet
that will connect to the endbricks outside these modules.

Note that our gadgetry is built on ideas as seen in the $W[1]$-hardness proof
of \IS in Unit Disk graphs; see Theorem~14.34 in \cite{fptbook}.

\paragraph*{Tuples and inequality propagation}

Without loss of generality, assume that each partition class $A_i$ has size
$n$.

In this proof, it is convenient to work with open boxes instead of closed
ones. Instead of using boxes as basic building blocks, we use \emph{box
tuples}. A box tuple consists of intersecting boxes, each of which is a
perturbed version of a single box, where the perturbations along all axes are
of the form $t/L^2$ for some $|t|\in [L/8]$. Most of our tuples will contain
$n$ boxes.

Clearly, an independent set selects at most one box from each box tuple. The
crucial property of a sequence of well-placed box tuples is that they can
express an inequality in the following sense. Suppose we have two box tuples
with axis $x_1$: in the first tuple the coordinates of box $i$ are
$(i/L^2,0,0)$, while in the second, the coordinates of box $i$ are
$(L+i/L^2,0,0)$. It is easy to see that any independent set of size $2$ will
have to select one box from each tuple. Moreover, if it selects box $i$ from
the first tuple and box $i'$ from the second, then $i\leq i'$ holds. In this
example, we say that the inequality is transferred through
$x_1$-perturbations.

\paragraph*{Construction overview}

Let $G,H$ be our input graphs, on $\bar{n}\bar{k}$ and $\bar{k}$ vertices
respectively. Our goal is to give a set of $poly(\bar{n})$ boxes that have an
independent set of size $k =
\Theta(\bar{k}^{\frac{d}{d-1}}/\alpha^{\frac{d}{d-1}})$ if and only if the input is
a \textsc{yes}-instance.

For a given vertex $h_i\in V(H)$, the subgraph isomorphism needs to map it to
some vertex $v_i\in A_i$. We can encode $v_i$ as an integer in $[n]$.

The task is to somehow check if a choice $v_i\in A_i$ and $v_j \in A_j$ is
valid, i.e., if there is an edge $(v_i,v_j)\in E(G)$. In other words, we can
encode the set of edges going between $A_i$ and $A_j$ as a set
$S_{i,j}\subseteq [n]^2$, and the task is to check if there exists an
$(a,b)\in S_{i,j}$ such that $v_i = a$ and $v_j=b$.

Let $L=\alpha^{\frac{d}{d-1}}$, and set $k \eqdef c \cdot
\bar{k}^{\frac{d}{d-1}}/L$, where $c$ will be specified later.
Let $\cC=\BEC^d(s,t)$, where $s \eqdef (k/t)^{1/d}$ and $t \eqdef (L/8)^{d-1} $.

We assign six
column-neighboring vertices of a cell in the bottom facet $P$ to each vertex $h\in
V(H)$, and five cross-arranged vertices of a cell in the top facet $Q$ of $\cC$ to
each edge of $H$. (In $3$ dimensions, cross-arranged vertices have indices
$(i,j),(i+1,j),(i-1,j),(i,j+1),(i,j-1)$.)

Our choice for the picture of $h\in V(H)$ is expressed as a number $n_h \in
[n]$, that is encoded as the index of the boxes within the independent set of
the box tuples assigned to $h$. We make sure that the box with the same index
is picked in each of these six tuples using an \emph{equalizer} gadget. For an
edge $hh'$, we cerate two wires starting from two of the box tuples
corresponding to $h$, one expressing $\leq n_h$, the other expressing $\geq
n_h$, to two opposite box tuples of the cross assigned to the edge $hh'$ (for
example, to the tuples $(i+1,j),(i-1,j)$). Similarly, we create two wires from
two of the box tuples corresponding to $h'$, expressing $\leq n_{h'}$, the
other expressing $\geq n_{h'}$, to the two other opposing tuples in the cross
(in our example, to $(i,j+1)$ and $(i,j-1)$). The middle of the cross, the box
tuple of index $(i,j)$ is replaced with an \emph{edge check} gadget.

These associations can be done injectively by a proper choice of $c$, since
\[|P|=|Q|=s^{d-1}t=(ck)^{\frac{d}{d-1}}t^{1/d}=\Theta(k^{\frac{d}{d-1}}\alpha) = \Theta(\bar{k}).\]
According
to Theorem~\ref{thm:becwire}, this wiring is realizable in $\cC$. The graph
induced by the wires has $\Theta(|V(\cC)|) =\Theta(s^d t)=\Theta(k)$ vertices.

\paragraph*{Gadgets: adapted gadgets, equalizer and edge check}

For the sake of simplicity, we describe our gadgets for $d=3$. A brick here
contains $(L/8)^2$ box tuples. To define a brick, we need to first define an
underlying "canonical" box for each tuple. These underlying boxes can have the same
perturbations as the ones allowed for box perturbations in the previous
construction. Secondly, we need to define the box perturbations within each of
the tuples $(i,j)\in [L/8]^2$ compared to this canonical box; these latter
perturbations we call offsets.

Observe that all the intersections used in our earlier gadgetry (with the
exception of parity fix gadgets that we do not use here) have the property
that they arise as two boxes touch at a facet, and the intersection is a
$(d-1)$-dimensional unit cube, contained within a $(d-1)$-dimensional
hyperplane. Therefore, we can generalize our bridge, elbow, adjust and
matching gadgets by introducing offsets perpendicular to the hyperplanes of
these intersections. The general matching gadget is constructed the same way
but without the parity fix gadgets.

\begin{figure}
\begin{center}
\includegraphics[height=5cm]{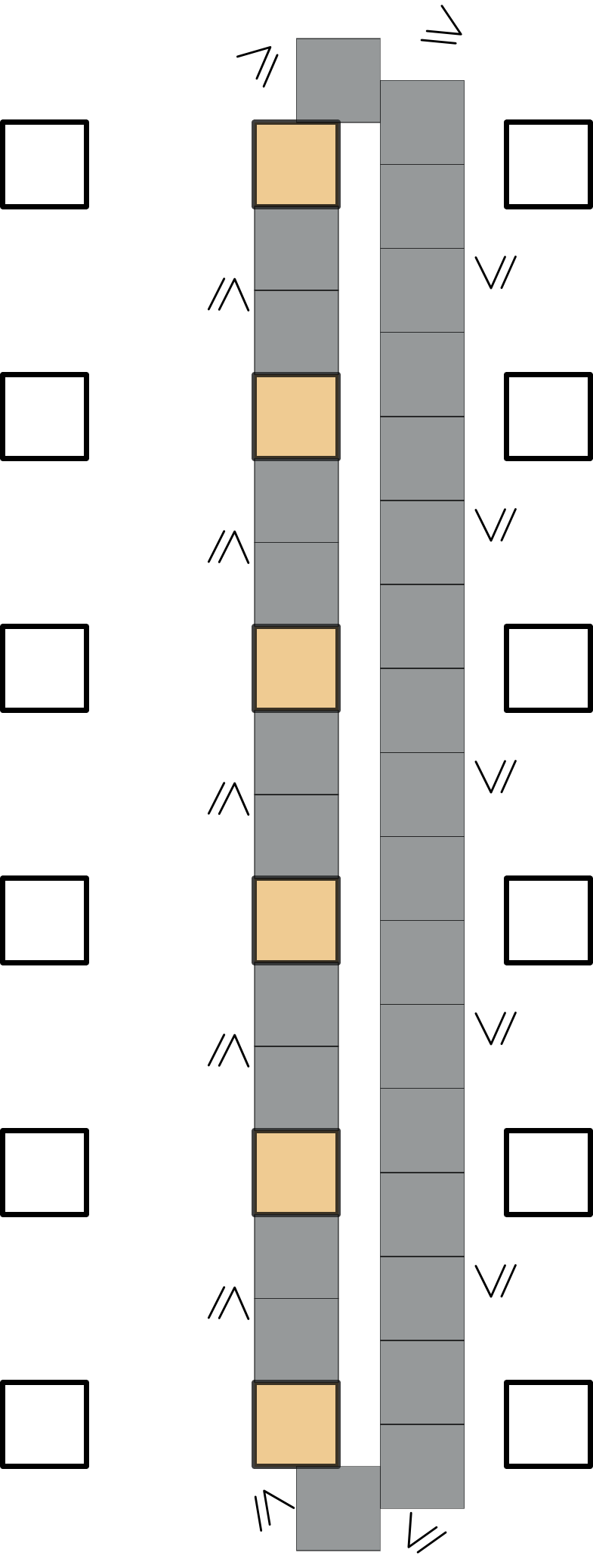}
\hspace{3cm}
\includegraphics[height=5cm]{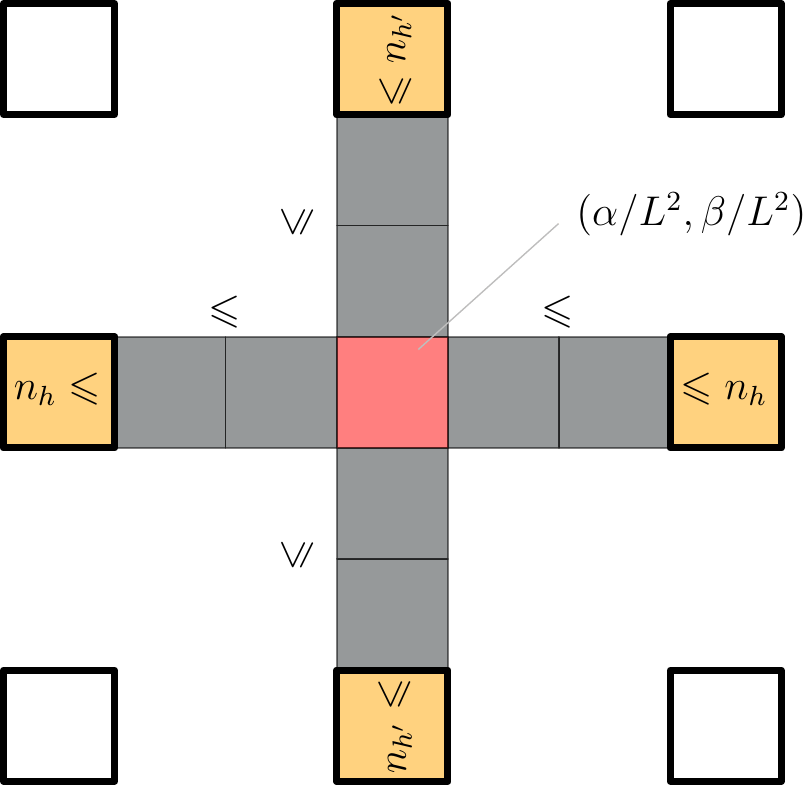}
\caption{Left: canonical boxes of an equalizer gadget. The offsets are defined
so that the indicated inequalities are realized. Right: canonical boxes of an
edge check gadget, with offsets along $x_1$ and $x_2$ corresponding to edges
between the partition classes to which $h$ and $h'$ are mapped
to.}\label{fig:equalizeedgcheck}
\end{center}
\end{figure}

Given our construction for the wiring from above, for each wire we can
introduce the offsets within each tuple so that the desired inequality is
carried through. Furthermore, we make sure that the start and end of each wire
is adjusted so that it can connect to the equalizer and edge check gadgets, as
detailed below.

The equalizer gadget relies on carrying an inequality along a cycle, see
Figure~\ref{fig:equalizeedgcheck} for a $3$-dimensional example, where the
view is from $x_3=\infty$. The tuples that connect to the individual wires are
drawn in orange. Notice that the offsets introduced on the $x_3$ coordinate
for the orange boxes must correspond to the type of inequality that the wire
carries.

The edge check gadget is again a simple construction, where the middle of the
cross is a tuple containing $O(n^2)$ boxes, where each box is associated with
an edge $e$ between the corresponding partition classes $A_k$ and $A_\ell$.
Such edges can be encoded as a subset of $[n]^2$, i.e., we have a box
$\alpha,\beta$ in the tuple if and only if vertex $\alpha$ of $A_k$ is
connected to vertex $\beta$ of $A_\ell$. The offset for box $(\alpha,\beta)$
compared to the canonical box in the brick is simply
$(0,\alpha/L^2,\beta/L^2)$.

This concludes the construction. Notice that the construction has an
independent set containing exactly one box from each box tuple if and only if
$H$ is a subgraph of $G$. The number of tuples needed for the equalizer and
edge check gadgets is insignificant compared to those needed in the wiring,
which is $\Theta(|V(\cC)|)=\Theta(k) = \Theta(\bar{k}^{\frac{d}{d-1}}/\alpha^{\frac{d}{d-1}})$. Consequently,
the construction has $k= \Theta(\bar{k}^{\frac{d}{d-1}}/\alpha^{\frac{d}{d-1}})$ box
tuples, as required.

The final instance has size $n=O(\bar{n}k)+O(\bar{k}\bar{n}^2)=O(\bar{n}^3)$. If for all $\gamma>0$ there is an algorithm for packing with
running time $n^{\gamma k^{1-1/d}\alpha}$, then we also have algorithms
for \textsc{Partitioned Subgraph Isomorphism} with running time
\[n^{\gamma k^{1-1/d}\alpha/\log k)}=\exp\left((\log \bar{n}) \cdot
\gamma' \left(\bar{k}^{\frac{d}{d-1}}/\alpha^{\frac{d}{d-1}}\right)^{1-1/d}\alpha/\log
\bar{k})\right) = \bar{n}^{\gamma'' \cdot \bar{k}/\log \bar{k}}\]
for all $\gamma''>0$, which would contradict
ETH.

\section{Conclusion}\label{sec:conclusion}

We have explored the impact of the stabbing number on the complexity of packing. We have seen that subexponential packing algorithms are possible for similarly sized objects if the stabbing number is $o(n^{1/d})$. The subexponential algorithms could be derived from powerful separator theorems, while the lower bounds required custom wiring results and non-trivial geometric gadgetry. We propose two open problems for future research.
\begin{itemize}
\item What is the precise impact of the stabbing number on the complexity of packing if objects are not similarly sized? One can get a subexponential algorithm by an adaptation of the separator in~\cite{BergBKMZ18}, but it yields an algorithm whose dependence on $\alpha$ is much weaker: it has $\alpha^d$ in the exponent instead of $\alpha$. Is this algorithm optimal?
\item Is there a subexponential algorithm for the \textsc{Dominating Set} problem in intersection graphs of $\alpha$-stabbed similarly sized objects? Or even for $n$ axis-parallel $1\times n^{\varepsilon}$ and $n^{\varepsilon} \times 1$ boxes in two dimensions?
\end{itemize}


\bibliography{boxes.bib}

 \end{document}